\documentclass[]{article}
\usepackage[english]{babel}
\usepackage{tikz}
\usepackage{amssymb}
\usepackage[mathscr]{eucal}
\usepackage[T1]{fontenc}
\usepackage[utf8]{inputenc}
\usepackage{calrsfs}
\usepackage{float}
\usepackage{amsmath}
\usepackage{amsthm}
\usepackage{enumerate}
\usepackage{caption}
\usepackage{multirow}
\usepackage[titletoc,toc]{appendix}
\usepackage{authblk}
\usetikzlibrary[shapes.misc]
\usetikzlibrary[shapes.geometric]
\newtheorem{defn}{Definition}[section]

\newtheorem{prop}[defn]{Proposition}
\newtheorem{cor}[defn]{Corollary}

\newtheorem{exmp}[defn]{Exemple}
\newtheorem{claim}[defn]{Claim}
\title{Subdivision into $i$-packings and $S$-packing chromatic number of some lattices}

\author[1,2]{Nicolas Gastineau\thanks{Author partially supported by the Burgundy Council}}
\author[2]{Hamamache Kheddouci}
\author[1]{Olivier Togni}
\affil[1]{\textit{Universit\'e de Bourgogne, 21078 Dijon cedex, France, Le2i, UMR CNRS 6303} }
\affil[2]{\textit{Université de Lyon, CNRS, Université Lyon 1, LIRIS, UMR5205, F-69622, France} }

\begin{document}
\maketitle
\begin{abstract}

An $i$-packing in a graph $G$ is a set of vertices at pairwise distance greater than $i$. For a nondecreasing sequence of integers $S=(s_{1},s_{2},\ldots)$, the $S$-packing chromatic number of a graph $G$ is the least integer $k$ such that there exists a coloring of $G$ into $k$ colors where each set of vertices colored $i$, $i=1,\ldots, k$, is an $s_i$-packing.
This paper describes various subdivisions of an $i$-packing into $j$-packings ($j>i$) for the hexagonal, square and
triangular lattices. These results allow us to bound the $S$-packing chromatic number for these graphs, with more precise
bounds and exact values for sequences $S=(s_{i}, i\in\mathbb{N}^{*})$, $s_{i}=d+ \lfloor (i-1)/n \rfloor$.
\end{abstract}
\section{Introduction}
Let $G=(V,E)$ be a (finite or infinite) graph and let $N(u)=\{v\in V(G) | \ uv\in E(G)\}$ be the set of neighbors of vertex $u$.
A set $X_{i}\subseteq V(G)$ is an $i$-packing if for any distinct pair $u$, $v$ $\in$ $X_{i}$, $d(u,v)> i$, where $d(u,v)$ denotes the usual shortest path distance between $u$ and $v$.
We will use $X_{i}$ to refer to an $i$-packing in a graph $G$.
A $k$-coloring $c$ of $G$ is a map from $V(G)$ to $\{1,\ldots, k\}$ such that for every pair of adjacent vertices $(u,v)$, we have $c(u)\neq c(v)$.
For a graph $G$ and a $k$-coloring $c$ of $G$, let $c_i$ be $\{u\in V(G) |\ c(u)=i\}$.
The smallest integer $k$ such that there exists a $k$-coloring of $G$ for which for every $i$, with $1\le i\le k$, $c_i$ is a $i$-packing, is called the \textit{packing chromatic number} of $G$ and is denoted by $\chi_{\rho}(G)$. This concept was introduced by Goddard \textit{et al.} \cite{GO2008} under the name of broadcast chromatic number.
More generally, for a nondecreasing sequence of integers $S=(s_{1},s_{2},\ldots)$, an \textit{$S$-packing $k$-coloring}
is a $k$-coloring $c$ of $V(G)$ such that for every $i$, with $1\le i\le k$, $c_i$ is a $s_{i}$-packing.
Such a coloring will also simply be called an $(s_1,\ldots,s_k)$-coloring. The {\em $S$-packing chromatic number} of $G$
denoted by $\chi^{S}_{\rho}(G)$, is the  smallest $k$ such that $G$ admits an $S$-packing $k$-coloring. For the
sequences $S=(s_{i}, i\in\mathbb{N}^{*})$, with $s_{i}=d+ \lfloor (i-1)/n \rfloor$, we call $\chi^{S}_{\rho}(G)$  the 
\textit{$(d,n)$-packing chromatic number} and denote it by $\chi^{d,n}_{\rho}(G)$.
For any connected graph $G$ such that $|V(G)|\ge d+1$, $\chi^{d,n}_{\rho}(G)\ge d+1$ and  $\chi^{1,1}_{\rho}(G)=\chi_{\rho}(G)$. For every bipartite graph $G$, $\chi^{1,2}_{\rho}(G)=2$ (a bipartite graph is 2-colorable).
Moreover, the smallest $n$ such that $\chi^{d,n}_{\rho}(G)= n$ corresponds to the $d$-distant chromatic number \cite{KR2008}, \textit{i.e.} the minimum number of $d$-packings that form a partition of the vertices.

Let $P_{\infty}$ denote the two-way infinite path, let $\mathbb{Z}^{2}=P_{\infty}\square P_{\infty}$ denote the planar square lattice (where $\square$ is the Cartesian product), $\mathscr{T}$ denote the planar triangular lattice and $\mathscr{H}$ denote the planar hexagonal lattice.
In this article, for an $(s_{1},s_{2},\ldots)$-coloring of a graph, we prefer to map vertices to the color multiset $\{s_{1},s_{2},\ldots\}$ even if two colors can then be denoted by the same number. This notation allows the reader to directly see to which type of packing the vertex belong depending on its color. When needed, we will denote colors of vertices in different $i$-packings by $i_a,i_b,\ldots$.
%If the number $i$ is present more than one time in the sequence $(s_{1},s_{2},\ldots)$, then the corresponding colors are denoted by $i_a,i_b,\ldots$}.
%Graphically, in order to represent a coloring, we will denote colors of vertices in different $i$-packings by $i_a,i_b,\ldots$.% or by $i$ followed by a letter of the alphabet ($2a$ and $2b$ for example) depending on the situation.

\subsection{Motivation and related work}
Packing colorings in graphs are inspired from frequency planning in wireless systems. The concept of $S$-packing
coloring emphasizes the fact that signals can have different powers similar to the packing coloring but enables the
presence of several signals with the same power, providing a more realistic model for the frequency assignment problem.

The packing chromatic number of lattices has been studied by several authors: Soukal and Holub \cite{SO2010} proved that
$\chi_{\rho}(\mathbb{Z}^{2})\le 17$, Ekstein \textit{et al.} \cite{EK2010} that
$12\le\chi_{\rho}(\mathbb{Z}^{2})$; Fiala \textit{et al.} \cite{FI2009} showed that $\chi_{\rho}(\mathscr{H})\le7$,
$\chi_{\rho}(\mathbb{Z}^{2}\square P_{2})=\infty$ and $\chi_{\rho}(\mathscr{H}\square P_{6})=\infty$ and Finbow and
Rall \cite{FIN2010} proved that $\chi_{\rho}(\mathscr{T})=\infty$.

$S$-packing colorings with sequences $S$ other than $(1,2,\ldots, k)$ first appear in \cite{GO2008,FI2010}. Goddard and Xu \cite{GOH2012} have recently studied $S$-packing colorings for the infinite path and for square and triangular lattices, determining conditions on the first elements of the sequence for which the graph is or is not $S$-packing-colorable.

Regarding the complexity, Goddard \textit{et al.} \cite{GO2008} proved that the problem of $(1,2,3)$-packing coloring is polynomial while $(1,2,3,4)$-packing coloring and $(1,2,2)$-packing coloring are NP-complete. Fiala and Golovach \cite{FI2010} showed that the problem of $(1,2,\ldots, k)$-coloring is NP-complete for trees. The NP-completeness of $(1,1,2)$-coloring was proved by Goddard and Xu \cite{GO2012} and afterward by Gastineau \cite{GA2015}.

While the packing coloring corresponds to an $S$-packing coloring with a strictly increasing sequence and the $d$-distant
chromatic number corresponds to a constant one, the sequence in the $(d,n)$-packing coloring also tends to infinity, but the parameter $n$ allows us to control its growth. 
% Let $G$ be a graph and $k$ be an integer, let $G^{k}$ denotes the graph G at the power $k$, i.e the graph $G$ with added edges between vertices at distance less than $k$.
% The $(d,n)$-packing chromatic number correspond to found the smallest $k=an+b$ with $0<b<n$ such that there exists $n$ sets $X^{d}_{i}$, $n$ sets $X^{d+1}_{i},\ldots ,n$ sets $X^{d+a}_{i}$, $b$ sets $X^{b+a+1}_{i}$, $i=1,\ \ldots,n$ that form a partition of $G$ and each $X^{j}_{i}$ are independent set in $G^{j}$.
% The principle of packing coloring comes from the area of frequency planning in wireless systems. The concept of $S$-packing coloring emphasizes the fact that signals can have different powers like with the packing coloring but enables several signals with the same power.\newline
% The complexity of $(d,n)$-packing chromatic number is known for $d=1$.
% Let $G$ be a graph. If $n=1$, Goddard \textit{et al.} \cite{GO2008} have proven that the problem of deciding whether $\chi_{\rho}(G)\le 3$ is polynomial and that the problem of deciding whether $\chi_{\rho}(G)=4$ is NP-hard.
% Moreover, the authors have shown that the problem of deciding whether a graph can be partitioned into two 1-packings and a 2-packing is also NP-hard.
% If $n\ge 3$, the problem of deciding whether $s_{n}=1$ and $\chi_{\rho}^{S}(G)=n$ is the problem of $n$-colorability ($n$COL), known to be NP-hard.\newline

Moreover, one can note that all the $S$-packing colorings of square and hexagonal lattices published so far have the property that the $s_1$-packing is
maximum and the other $s_i$-packings are obtained by subdivisions of $s_1$-packings (and are not always maximum).
Therefore, we find it interesting to study subdivision of an $i$-packing into $j$-packings, $j>i$, in lattices. These
subdivisions can in turn be used to describe patterns to obtain an $S$-packing coloring of a lattice.
However, determining the families of graphs $G$ for which for any $S$ such that $G$ is $S$-colorable, the $S$-coloring satisfies the above property is satisfied remains an open question.
Recently Goddard and Xu \cite{GOH2012} proved that there exist nondecreasing sequences $S$ such that $P_{\infty}$ is $S$-colorable and in any $(s_1,\ldots,\chi_{\rho}^{S}(P_\infty))$-packing coloring of $P_{\infty}$, the $s_1$-packing is not maximum, showing that for $P_{\infty}$, there are sequences $S$ for which the above property is not satisfied.

%In order to find an optimal $(d,n)$-packing coloring, one of the main question will be how to subdivide a $d$-packing.
% Since in the patterns used to find upper bounds for the different lattices there are $j$-packings used as $i$-packings, $j>i$ (like the 2-packing in $\mathbb{H}$ used as a 3-packing), our idea is to find some properties of decomposition of an $i$-packing into $j$-packings, $j>i$, which can describe patterns to color a lattice with $i$-packings.
% In most cases, in a optimal $(d,n)$-packing coloring, the $d$-packings in the partition correspond to $d$-packings in an optimal $d$-distant coloring and the remaining $d$-packings are subdivided into $(d+i)$-packings.

\subsection{Our results}
% The aim of this paper is to find lower and upper bounds for $S$-packing chromatic numbers of lattices which have a finite packing chromatic number like $\mathbb{H}$ and $\mathbb{Z}^{2}$ but also of lattices which have an infinite packing chromatic number like $\mathbb{T}$.
The second section introduces some definitions and results related to density. The third section introduces some
subdivision of the lattices into $i$-packings. The fourth and fifth sections give lower bounds resulting from Section 2
and upper bounds resulting from Section 3 for the $S$-packing chromatic number and the $(d,n)$-packing chromatic number
of the lattices $\mathscr{H}$, $\mathbb{Z}^{2}$ and $\mathscr{T}$.
Tables~\ref{tab1} ,~\ref{tab2} and~\ref{tab3} summarize the values obtained in this paper for the $(d,n)$-packing chromatic number, giving an idea of our results. The emphasized numbers are exact values and all pairs of values are lower and upper bounds. Lower bounds have been calculated from Proposition 2.2 and Propositions 2.5, 2.7 and 2.9. Some of the results for square and triangular lattice have been found independently by Goddard and Xu \cite{GOH2012}. 
\begin{table}
\begin{center}
\begin{tabular}{| c |  c | c | c | c | c | c |}
\hline
$d\backslash n$& 1 & 2 & 3 & 4 & 5 & 6 \\ \hline
1 & \textbf{7}\cite{AV2007,FI2009}  & \textbf{2} & \textbf{2} & \textbf{2} & \textbf{2} & \textbf{2} \\
2 & $\infty$ & 5 - 8 & \textbf{5} & \textbf{4} & \textbf{4} & \textbf{4} \\
3 & $\infty$ & 15 - 35 & 9 - 13 & 8 - 10 & 7 - 8 & \textbf{6} \\
4 & $\infty$ & 61 - ? & 20 - 58 & 15 - 27 & 13 - 21 &  12 - 18 \\
5 & $\infty$ & $\infty$   & 37 - ? & 25 - ? & 21 - ? & 19 - ? \\
8 & $\infty$ & $\infty$   & $\infty$ & ? & ? & ? \\
11 & $\infty$ & $\infty$   & $\infty$ & $\infty$ & ? & ? \\
13 & $\infty$ & $\infty$   & $\infty$ & $\infty$ & $\infty$ & ? \\
16 & $\infty$ & $\infty$   & $\infty$ & $\infty$ & $\infty$ & $\infty$ \\  \hline
\end{tabular}
\caption{Bounds for $(d,n)$-packing chromatic numbers of the hexagonal lattice.}
\label{tab1}
\end{center}
\end{table}
\begin{table}
\begin{center}
\begin{tabular}{| c |  c | c | c | c | c | c |}
\hline
$d\backslash n$& 1 & 2 & 3 & 4 & 5 & 6 \\ \hline
1 & 12 - 17 \cite{EK2010,SO2010} & \textbf{2} & \textbf{2} & \textbf{2} & \textbf{2} & \textbf{2} \\
2 & $\infty$ & 11 - 20 & 7 - 8 & \textbf{6} \cite{GOH2012} & \textbf{5} \cite{GOH2012} & \textbf{5} \\
3 & $\infty$ & 57 - ? & 16 - 33 & 12 - 20 & 10 - 17 & 10 - 14 \\
4 & $\infty$ & $\infty$ & 44 - ? & 25 - 56 & 20 - 34 & 18 - 28 \\
5 & $\infty$ & $\infty$ & 199 - ? & 50 - ? & 35 - ? & 29 - ? \\
6 & $\infty$ & $\infty$ & $\infty$ & ? & ? & ? \\
8 & $\infty$ & $\infty$ & $\infty$ & $\infty$ & ? & ? \\
10 & $\infty$ & $\infty$   & $\infty$ & $\infty$ & $\infty$ & ? \\ 
12 & $\infty$ & $\infty$   & $\infty$ & $\infty$ & $\infty$ &$\infty$ \\ \hline
\end{tabular}
\caption{Bounds for $(d,n)$-packing chromatic numbers of the square lattice.}
\label{tab2}
\end{center}
\end{table}
\begin{table}
\begin{center}
\begin{tabular}{| c |  c | c | c | c | c | c |}
\hline
$d\backslash n$& 1 & 2 & 3 & 4 & 5 & 6 \\ \hline
1 & $\infty$\cite{FIN2010} & 5 - 6 \cite{GOH2012} & \textbf{3} & \textbf{3} & \textbf{3} & \textbf{3} \\
2 & $\infty$ & 127  - ? & 14 - ? & 10 - 16 & 9 - 13 & 8 - 10 \\
3 & $\infty$ & $\infty$ & 81 - ? & 28 - 72 & 20 - 38 & 17 - 26 \\
4 & $\infty$ & $\infty$ & $\infty$ & 104 - ? & 49 - ? & 36 - ? \\
5 & $\infty$ & $\infty$ & $\infty$ & $\infty$ & ? & ? \\
7 & $\infty$ & $\infty$ & $\infty$ & $\infty$ & $\infty$ & ? \\ 
8 & $\infty$ & $\infty$ & $\infty$ & $\infty$ & $\infty$ & $\infty$ \\ \hline
\end{tabular}
\caption{Bounds for $(d,n)$-packing chromatic numbers of the triangular lattice.}
\label{tab3}
\end{center}
\end{table}
\section{Density of $i$-packings}
\subsection{Density of an $i$-packing in a lattice}
Let $G=(V,E)$ be a graph, finite or infinite and let $n$ be a positive integer.
For a vertex $x$ of G, the ball of radius $n$ centered at $x$ is the set $B_{n}(x)=\{v\in V(G)| d_{G}(x,v)\le n\}$ and the sphere of radius $n$ centered at $x$ is the set  $\partial B_{n}(x)=\{v\in V(G)| d_{G}(x,v)= n\}$.
The density of a set of vertices $X\subseteq V(G)$ is $d(X)=\limsup\limits_{l\to \infty}  \max\limits_{x\in V}\{ \frac{|X\cap B_{l}(x)|}{| B_{l}(x)|}\}$.\newline
The notion of $k$-area was introduced by Fiala \textit{et al.} \cite{FI2009}. We propose here a slightly modified definition:
\begin{defn}
Let $G$ be a graph, $x\in V(G)$, and let $k$ be a positive integer.
The $k$-area $A(x,k)$ assigned to $G$ is defined by :\newline
$A(x,k)=\left\{
    \begin{array}{ll}
        |B_{k/2}(x)| & \mbox{for $k$ even;} \\
        |B_{\lfloor k/2 \rfloor}(x)|+\sum\limits_{\substack{u\in\partial B_{\lceil k/2 \rceil}(x)}}\frac{|N(u)\cap B_{\lfloor k/2 \rfloor}(x)|+ |N(u)\cap \partial B_{\lceil k/2 \rceil}(x)|/2}{deg(u)} 
& \mbox{for $k$ odd.}
    \end{array}
\right.$
\end{defn}
For vertex-transitive graphs, the k-areas are the same for all vertices, hence we denote it by $A(k)$.\newline
The motivation for our modification of the notion of $k$-area with the introduction of the set of neighbors inside the sphere is to have sharper density bounds than the ones obtained by the initial notion of $k$-area. For the square and the hexagonal lattice the notion coincide as the relation is empty. However, for the triangular lattice, the density bound is smaller: the definition of Fiala \textit{et al.} \cite{FI2009} gives $A(1)=2$ whereas $A(1)=3$ in our case since there are adjacent vertices in the sphere (as for every $u\in\partial B_{1}(x)$, $|N(u)\cap \partial B_{1}(x)|/2=1$, then it adds one to the initial definition of $k$-area).
Figure 1 illustrates this example giving a coverage of the triangular lattice by balls of radius $1$. In one case (on the left) the balls are disjoint and in the second case (on the right) each sphere can be shared by several balls. Observe that in the second case, each vertex $u$ in a sphere centered at $x$ has two neighbors, and hence $|N(u)\cap \partial B_{1}(x)|/2=1$.
\begin{figure}[t]
\begin{center}
\begin{tikzpicture}
\foreach \x in {-0.5,0,0.5,...,4.5}
{
\draw (\x,0) -- (\x,3);
}
\foreach \y in {0,0.5, ...,3}
{
\draw (-0.5,\y) -- (4.5,\y);
}
\foreach \z in {-0.5,0,0.5,...,1.5}
{
\draw (\z,3) -- (\z+3,0);
}
\foreach \h in {0.5,1,...,2.5}
{
\draw (-0.5,\h) -- (\h-0.5,0);
\draw (\h+1.5,3) -- (4.5,\h);
}
\foreach \u in {0.5,1,...,2}
{
\node at (2*\u,\u) [circle,draw=blue!50,fill=blue!20] {};
\draw[line width=2pt,color=red] (2*\u-0.5,\u) -- (2*\u-0.5,\u+0.5);
\draw[line width=2pt,color=red] (2*\u,\u+0.5) -- (2*\u-0.5,\u+0.5);
\draw[line width=2pt,color=red] (2*\u,\u+0.5) -- (2*\u+0.5,\u);
\draw[line width=2pt,color=red] (2*\u+0.5,\u-0.5) -- (2*\u+0.5,\u);
\draw[line width=2pt,color=red] (2*\u+0.5,\u-0.5) -- (2*\u,\u-0.5);
\draw[line width=2pt,color=red] (2*\u,\u-0.5) -- (2*\u-0.5,\u);
}
\foreach \u in {0,0.5}
{
\node at (2*\u+3.5,\u) [circle,draw=blue!50,fill=blue!20] {};
\draw[line width=2pt,color=red] (2*\u+3,\u) -- (2*\u+3,\u+0.5);
\draw[line width=2pt,color=red] (2*\u+3.5,\u+0.5) -- (2*\u+3,\u+0.5);
}
\foreach \u in {0.5,1}
{
\node at (2*\u-0.5,\u+1.5) [circle,draw=blue!50,fill=blue!20] {};
\draw[line width=2pt,color=red] (2*\u-1,\u+1.5) -- (2*\u-1,\u+2);
\draw[line width=2pt,color=red] (2*\u-0.5,\u+2) -- (2*\u-1,\u+2);
\draw[line width=2pt,color=red] (2*\u-0.5,\u+2) -- (2*\u,\u+1.5);
\draw[line width=2pt,color=red] (2*\u,\u+1) -- (2*\u,\u+1.5);
\draw[line width=2pt,color=red] (2*\u,\u+1) -- (2*\u-0.5,\u+1);
\draw[line width=2pt,color=red] (2*\u-0.5,\u+1) -- (2*\u-1,\u+1.5);
}
\node at (0,0) [circle,draw=blue!50,fill=blue!20] {};
\node at (-0.5,1.5) [circle,draw=blue!50,fill=blue!20] {};
\node at (2.5,3) [circle,draw=blue!50,fill=blue!20] {};
\draw[line width=2pt,color=red] (-0.5,0) -- (-0.5,0.5);
\draw[line width=2pt,color=red] (-0.5,0.5) -- (0,0.5);
\draw[line width=2pt,color=red] (0,0.5) -- (0.5,0);
\draw[line width=2pt,color=red] (4,0) -- (3.5,0.5);
\draw[line width=2pt,color=red] (4,0.5) -- (4.5,0);
\draw[line width=2pt,color=red] (-0.5,1) -- (0,1);
\draw[line width=2pt,color=red] (0,1) -- (0,1.5);
\draw[line width=2pt,color=red] (0,1.5) -- (-0.5,2);
\draw[line width=2pt,color=red] (2,3) -- (2.5,2.5);
\draw[line width=2pt,color=red] (2.5,2.5) -- (3,2.5);
\draw[line width=2pt,color=red] (3,2.5) -- (3,3);

\foreach \x in {5,5.5,6,...,10}
{
\draw (\x,0) -- (\x,3);
}
\foreach \y in {0,0.5, ...,3}
{
\draw (5,\y) -- (10,\y);
}
\foreach \z in {5,5.5,...,7}
{
\draw (\z,3) -- (\z+3,0);
}
\foreach \h in {6,6.5,...,8}
{
\draw (5,\h-5.5) -- (\h-0.5,0);
\draw (\h+1.5,3) -- (10,\h-5.5);
}
\foreach \u in {0,0.5,...,3}
\foreach \v in {5.5,7}
{
\node at (\u+\v,\u) [circle,draw=blue!50,fill=blue!20] {};
}
\foreach \u in {0,0.5,...,1.5}
\foreach \v in {8.5}
{
\node at (\u+\v,\u) [circle,draw=blue!50,fill=blue!20] {};
}
\foreach \u in {1,1.5,...,3}
\foreach \v in {4}
{
\node at (\u+\v,\u) [circle,draw=blue!50,fill=blue!20] {};
}
\node at (10,0) [circle,draw=blue!50,fill=blue!20] {};
\node at (5,2.5) [circle,draw=blue!50,fill=blue!20] {};
\node at (5.5,3) [circle,draw=blue!50,fill=blue!20] {};
\foreach \u in {0,1.5}
\foreach \v in {0,1.5,3}
{
\draw[line width=2pt,color=red] (6+\v,\u) -- (6.5+\v,\u);
\draw[line width=2pt,color=red] (6+\v,\u) -- (5.5+\v,\u+0.5);
\draw[line width=2pt,color=red] (6.5+\v,\u) -- (6.5+\v,\u+0.5);
\draw[line width=2pt,color=red] (5+\v,\u+0.5) -- (5.5+\v,\u+0.5);
\draw[line width=2pt,color=red] (5.5+\v,\u+0.5) -- (5.5+\v,\u+1);
\draw[line width=2pt,color=red] (6+\v,\u+1) -- (6.5+\v,\u+0.5);
\draw[line width=2pt,color=red] (5.5+\v,\u+1) -- (6+\v,\u+1);
\draw[line width=2pt,color=red] (6+\v,\u+1) -- (6+\v,\u+1.5);
\draw[line width=2pt,color=red] (5.5+\v,\u+1) -- (5+\v,\u+1.5);
}
\draw[line width=2pt,color=red] (5,0) -- (5,0.5);
\draw[line width=2pt,color=red] (5,1.5) -- (5,2);
\draw[line width=2pt,color=red] (6,3) -- (6.5,3);
\draw[line width=2pt,color=red] (7.5,3) -- (8,3);
\draw[line width=2pt,color=red] (9,3) -- (9.5,3);
\draw[line width=2pt,color=red] (9.5,3) -- (10,2.5);
\draw[line width=2pt,color=red] (9.5,1.5) -- (10,1);
\draw[line width=2pt,color=red] (9.5,0.5) -- (10,0.5);
\draw[line width=2pt,color=red] (9.5,2) -- (10,2);
\draw[line width=2pt,color=red] (10,0.5) -- (10,1);
\draw[line width=2pt,color=red] (10,2) -- (10,2.5);
\node at (2,-0.5){A(2)=7};
\node at (7.5,-0.5){A(1)=3};
\end{tikzpicture}
\end{center}
\caption{Examples of $k$-area in $\mathscr{T}$.}
\end{figure}
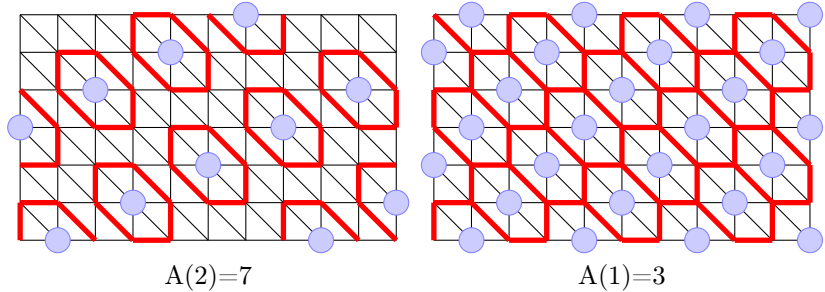

\begin{prop}
Let $G$ be a vertex-transitive graph with finite degree, and $i$ be a positive integer. If $X_{i}$ is an $i$-packing in $G$, then
$$d(X_{i})\le \frac{1}{A(i)}.$$
\end{prop}
\begin{proof}
Observe that for arbitrary vertices $x$ and $y$ of an $i$-packing $X_{i}$, the sets $B_{\lfloor i/2 \rfloor}(x)$ and $B_{\lfloor i/2 \rfloor}(y)$ are disjoint, since the vertices $x$ and $y$ are at distance greater than $i$. Then $d(X_{i})<1/|B_{\lfloor i/2 \rfloor}(x)|$. Assume that $i$ is an odd number and let $u$ be a vertex at distance $\lceil i/2 \rceil$ from $x$, then $u$ has $\text{deg}(u)$ neighbors. If among these $\text{deg}(u)$ neighbors, $k$ neighbors are in $B_{\lfloor i/2 \rfloor}(x)$ then $u$ can be at distance $\lceil i/2 \rceil$ from only $k/\text{deg}(u)$ vertices in $X_{i}$. Hence $d(X_{i})<1/(|B_{\lfloor i/2 \rfloor}(x)|+\sum\limits_{\substack{u\in\partial B_{\lceil k/2 \rceil}(x)}}\frac{|N(u)\cap B_{\lfloor k/2 \rfloor}(x)|}{deg(u)} ) $.\newline
Moreover $u$ and a neighbor $v$ of $u$ in $B_{\lfloor i/2 \rfloor}(x)$ cannot be both at distance $\lceil i/2 \rceil$ from more than 2 vertices in $X_{i}$, therefore $uv$ can only belong to two spheres of radius $\lceil i/2 \rceil$ centered at a vertex in $X_{i}$. Hence it follows that $d(X_{i})<1/A(i)$.
\end{proof}
\begin{cor}
Let $G$ be a vertex-transitive graph with finite degree, and $i$ be a positive integer. If $G$ has a finite $S$-packing chromatic number, then
$$\sum^{\infty}_{i=1}\frac{1}{A(s_{i})}\ge 1.$$
\end{cor}
\begin{cor}
Let $G$ be a vertex-transitive graph with finite degree, and $i$ be a positive integer. If $G$ has a finite $(d,n)$-packing chromatic number, then
$$\sum^{\infty}_{i=d}\frac{n}{A(i)}\ge 1.$$
\end{cor}
An $i$-packing $X_{i}$ is called a \textit{maximized $i$-packing} if for any other $i$-packing $X^{'}_{i}$, $d(X_{i})\ge d(X^{'}_{i})$.
\subsection{Density of an $i$-packing in the hexagonal lattice}
\begin{prop}
Let $\mathscr{H}$ be the hexagonal lattice, $x$ be a vertex in $V(\mathscr{H})$ and $n$ be a positive integer. Then
\begin{enumerate}
\setlength{\parskip}{.05cm}
\item $|\partial B_{n}(x)|=3n$;
\item $|B_{n}(x)|=\frac{3}{2}n^{2}+\frac{3}{2}n +1$.
\end{enumerate}
\label{prop1}
\end{prop}
\begin{proof}
1. As the set $\partial B_{n}(x)$ always contains three more vertices than $\partial B_{n-1}(x)$, then $|\partial B_{n}(x)|=3n$.\newline
2. The graph $\mathscr{H}$ is 3-regular and so $|B_{1}(x)|=4$. Suppose the statement is true for $n$, then $|B_{n+1}(x)|=|B_{n}(x)|+|\partial B_{n+1}(x)|=\frac{3}{2}n^{2}+\frac{3}{2}n +1+3(n+1)=(\frac{3}{2}n^{2}+\frac{3}{2}+3n)+(\frac{3}{2}n+\frac{3}{2})+1=\frac{3}{2}(n+1)^{2}+\frac{3}{2}(n+1) +1$
and the result follows by induction.
\end{proof}
\begin{prop}
Let $\mathscr{H}$ be the hexagonal lattice and $k$ be a positive integer. Then
\begin{enumerate}
\setlength{\parskip}{.05cm}
\item $A(2k)=\frac{3}{2}k^{2}+\frac{3}{2}k +1$;
\item $A(4k+1)=6k^{2}+6k +2$;
\item $A(4k+3)=6k^{2}+12k +6$.
\end{enumerate}
\end{prop}
\begin{proof}
1. The first property results easily from Proposition \ref{prop1}.\newline
2. If $n=4k+1$, then Proposition \ref{prop1} gives $|B_{2k}(x)|=\frac{3}{2}(2k)^{2}+\frac{3}{2}(2k) +1=6k^{2}+3k+1$.
For every vertex $y$ in $\partial B_{2k+1}(x)$, $y$ has no neighbor in $\partial B_{2k+1}(x)$ other than itself, so $|N(y)\cap \partial B_{2k+1}(x)|=0$.
We have to distinguish two kinds of vertices: $3k$ vertices have two neighbors in $B_{2k}(x)$ and $|\partial B_{2k+1}(x)|-3k=3k+3$ vertices have one neighbor in $B_{2k}(x)$.
Therefore, $|A(4k+1)|=6k^{2}+3k+1+\frac{6k}{3}+\frac{3k+3}{3}=6k^{2}+6k +2$.\newline
3. If $n=4k+3$, then Proposition \ref{prop1} gives $|B_{2k+1}(x)|=\frac{3}{2}(2k+1)^{2}+\frac{3}{2}(2k+1) +1=6k^{2}+9k+4$.
For every vertex $y$ in $\partial B_{2k+2}(x)$, $y$ has no neighbor in $\partial B_{2k+2}(x))$ other than itself, so $|N(y)\cap \partial B_{2k+2}(x)|=0$.
We have to distinguish two kinds of vertices: $3k$ vertices have two neighbors in $B_{2k+2}(x)$ and $|\partial B_{2k+2}(x)|-3k=3k+6$ vertices have one neighbor in $B_{2k}(x)$.
Hence, we have $|A(4k+3)|=6k^{2}+9k+4+\frac{6k}{3}+\frac{3k+6}{3}=6k^{2}+12k +6$.
\end{proof}
Note that this result appeared in the article of Goddard and Xu \cite{GOH2012}.
\subsection{Density of an $i$-packing in the square lattice}
\begin{prop}
Let $\mathbb{Z}^{2}$ be the square lattice, $x$ be a vertex in $V(\mathbb{Z}^{2})$ and $n$ be a positive integer. Then
\begin{enumerate}
\setlength{\parskip}{.05cm}
\item $|\partial B_{n}(x)|=4n$;
\item $|B_{n}(x)|=2 n^{2}+2 n +1$.
\end{enumerate}
\label{prop2}
\end{prop}
\begin{proof}
1. As the set $\partial B_{n}(x)$ always contains four more vertices than $\partial B_{n-1}(x)$, then $|\partial B_{n}(x)|=4n$.\newline
2. The graph $\mathbb{Z}^{2}$ is 4-regular and so $|B_{1}(x)|=5$. Suppose the statement is true for $n$, then $|B_{n+1}(x)|=|B_{n}(x)|+|\partial B_{n+1}(x)|=2 n^{2}+2 n +1+4n+4= 2 (n+1)^{2}+ 6n+5-4n-2= 2 (n+1)^{2} +2(n+1)+1$
and the result follows by induction.
\end{proof}
\begin{prop}
Let $\mathbb{Z}^{2}$ be the square lattice and $k$ be a positive integer. Then
\begin{enumerate}
\setlength{\parskip}{.05cm}
\item $A(2k)=2k^{2}+2k+1$;
\item $A(2k+1)=2k^{2}+4k +2$.
\end{enumerate}
\end{prop}
\begin{proof}
1. The first property results easily from Proposition \ref{prop2}.\newline
2. If $n=2k+1$, then Proposition \ref{prop2} gives $|B_{k}(x)|=k^{2}+2k +1$.
For every vertex $y$ in $\partial B_{k+1}(x)$, $y$ has no neighbor in $\partial B_{k+1}(x)$ other than itself, so $|N(y)\cap \partial B_{k+1}(x)|=0$.
We have to distinguish two kinds of vertices: $4k$ vertices have two neighbors in $B_{k}(x)$ and 4 vertices have one neighbor in $B_{k}(x)$.
Hence, we have $|A(2k+1)|=2k^{2}+2k+1+2\frac{4k}{4}+\frac{4}{4}=2k^{2}+4k +2$.
\end{proof}
Note that this result appeared implicitly in the article of Fiala et al. \cite{FI2009}.
\subsection{Density of an $i$-packing in the triangular lattice}
\begin{prop}
Let $\mathscr{T}$ be the triangular lattice, $x$ be a vertex in $V(\mathscr{T})$ and $n$ be a positive integer. Then
\begin{enumerate}
\setlength{\parskip}{.05cm}
\item $|\partial B_{n}(x)|=6n$;
\item $|B_{n}(x)|=3 n^{2}+3 n +1$.
\end{enumerate}
\label{prop3}
\end{prop}
\begin{proof}
1. As the set $\partial B_{n}(x)$ always contains six more vertices than $\partial B_{n-1}(x)$, then $|\partial B_{n}(x)|=6n$.\newline
2. The graph $\mathscr{T}$ is 6-regular and so $|B_{1}(x)|=7$. Suppose the statement is true for $n$, then $|B_{n+1}(x)|=|B_{n}(x)|+|\partial B_{n+1}(x)|=3 n^{2}+3 n +1+6n+6= 3 (n^{2}+1)+ 3n+1+6-3= 3 (n^{2}+1)+3(n+1)+1$
and the result follows by induction.
\end{proof}
\begin{prop}
Let $\mathscr{T}$ be the triangular lattice and $k$ be a positive integer. Then
\begin{enumerate}
\setlength{\parskip}{.05cm}
\item $A(2k)=3k^{2}+3k+1$;
\item $A(2k+1)=3k^{2}+6k +3$.
\end{enumerate}
\end{prop}
\begin{proof}
1. The first property result easily from Proposition \ref{prop3}.\newline
2. If $n=2k+1$, then Proposition \ref{prop3} gives $|B_{k}(x)|=3k^{2}+3k +1$.
For every vertex $y$ in $\partial B_{k+1}(x)$, $y$ has two neighbors in $\partial B_{k+1}(x))$ other than itself, so $|N(y)\cap \partial B_{2k+1}(x)|=2$.
We have to distinguish two kinds of vertices: $6k$ vertices have two neighbors in $B_{k}(x)$ and six vertices have one neighbor in $B_{k}(x)$.
Hence, we have $|A(2k+1)|=3k^{2}+3k+1+\frac{6k+6}{6}+2\frac{6k}{6}+\frac{6}{6}=3k^{2}+6k +3$.
\end{proof}
\section{Subdivision of an $i$-packing in $\mathscr{H}$, $\mathbb{Z}^{2}$ and $\mathscr{T}$}
\subsection{Subdivision of a 2-packing in $\mathscr{H}$}
\begin{figure}[t]
\begin{center}
\begin{tikzpicture}
\foreach \x in {0,1,...,2}
\foreach \y in {0,1}
{
\draw (\x,\y) -- (\x-0.5,\y);
\draw (\x-0.5,\y) --(\x-0.5,\y+0.5);
\draw (\x-0.5,\y+0.5) -- (\x,\y+0.5);
\draw (\x,\y+0.5) -- (\x,\y+1);
\draw (\x,\y) -- (\x+0.5,\y);
\draw (\x+0.5,\y) -- (\x+0.5,\y+0.5);
\draw (\x+0.5,\y+0.5) -- (\x,\y+0.5);
\draw (\x,\y+0.5) -- (\x,\y+1);
}
\foreach \w in {0,1}
{
\draw (0.5+\w,2) -- (0.5+\w,2.5);
}
\draw (-0.5,2) --(2.5,2);
\foreach \v in {0,1,2}
\foreach \u in {0,2}
{
\node at (\u,\v) [circle,draw=blue!50,fill=blue!20] {};
\node at (\u+0.3,\v+0.3){2};
}
\node at (1,0.5) [circle,draw=blue!50,fill=blue!20] {};
\node at (1.3,0.8){2};
\node at (1,1.5) [circle,draw=blue!50,fill=blue!20] {};
\node at (1.3,1.8){2};

\foreach \x in {3.5,4.5,5.5, 6.5}
\foreach \y in {0,1}
{
\draw (\x,\y) -- (\x-0.5,\y);
\draw (\x-0.5,\y) --(\x-0.5,\y+0.5);
\draw (\x-0.5,\y+0.5) -- (\x,\y+0.5);
\draw (\x,\y+0.5) -- (\x,\y+1);
\draw (\x,\y) -- (\x+0.5,\y);
\draw (\x+0.5,\y) -- (\x+0.5,\y+0.5);
\draw (\x+0.5,\y+0.5) -- (\x,\y+0.5);
\draw (\x,\y+0.5) -- (\x,\y+1);
}
\foreach \w in {0,1,2}
{
\draw (4+\w,2) -- (4+\w,2.5);
}
\draw (3,2) --(7,2);
\foreach \u in {3.5,6.5}
\foreach \v in {0,1,2}
{
\node at (\u,\v) [circle,draw=red!50,fill=red!20] {};
\node at (\u+0.3,\v+0.3){3};
}
\node at (5,0.5) [circle,draw=red!50,fill=red!20] {};
\node at (5.3,0.8){3};
\node at (5,1.5) [circle,draw=red!50,fill=red!20] {};
\node at (5.3,1.8){3};

\foreach \x in {8,9,...,12}
\foreach \y in {0,1}
{
\draw (\x,\y) -- (\x-0.5,\y);
\draw (\x-0.5,\y) --(\x-0.5,\y+0.5);
\draw (\x-0.5,\y+0.5) -- (\x,\y+0.5);
\draw (\x,\y+0.5) -- (\x,\y+1);
\draw (\x,\y) -- (\x+0.5,\y);
\draw (\x+0.5,\y) -- (\x+0.5,\y+0.5);
\draw (\x+0.5,\y+0.5) -- (\x,\y+0.5);
\draw (\x,\y+0.5) -- (\x,\y+1);
}
\foreach \w in {0,1,2,3}
{
\draw (8.5+\w,2) -- (8.5+\w,2.5);
}
\draw (7.5,2) --(12.5,2);
\node at (8,0) [circle,draw=green!50,fill=green!20] {};
\node at (8.3,0.3){4};
\node at (7.5,1.5) [circle,draw=green!50,fill=green!20] {};
\node at (7.8,1.8){4};
\node at (9.5,1) [circle,draw=green!50,fill=green!20] {};
\node at (9.8,1.3){4};
\node at (11.5,0.5) [circle,draw=green!50,fill=green!20] {};
\node at (11.8,0.8){4};
\node at (11,2) [circle,draw=green!50,fill=green!20] {};
\node at (11.3,2.3){4};
\end{tikzpicture}
\end{center}
\caption{The sets $X_{2}$ (2-packing), $X_{3}$ (3-packing) and $X_{4}$ (4-packing) in $\mathscr{H}$.}
\end{figure}
Let $X_{2}$ be the (unique) maximized 2-packing in $\mathscr{H}$ represented in Figure 2. Note that $d(X_{2})=1/A(2)=1/4$ and remark that four 2-packings form a partition of $\mathscr{H}$ if we translate $X_{2}$ three times.\newline
The hexagonal lattice can be seen as a subgraph of the square lattice. In fact in Figure 2, $\mathscr{H}$ is represented as subgraph of the usual representation of the square lattice. In the square lattice, we can choose one vertex as the origin and all the other vertices can be nominated by a Cartesian coordinate. In every description of $\mathscr{H}$, our origin $(0,0)$ will be a vertex in the packing that we want to describe such that there is no edge between $(0,0)$ and $(0,1)$. In fact we illustrate packings with a figure in this subsection but it will not be the case after; we will use Cartesian coordinates in order to describe a packing.
For example, $X_{2}$ from Figure 2 is the set of vertices: $X_{2}=\{(2x+4y,x)|\ x\in \mathbb{Z},\ y\in \mathbb{Z} \}$.\newline
In Appendix A, we recall a proposition about distance in the hexagonal lattice from Jacko and Jendrol \cite{JA2005}.
This proposition is useful to verify that a set is an $i$-packing. These verifications are left to the reader in the remaining propositions.

\begin{prop}
Let $k>0$ and $m>0$ be integers. There exist:

   \begin{enumerate}
\setlength{\parskip}{.05cm}
	\item $k^{2}$ $(3k-1)$-packings that form a partition of $X_{2}$;
	\item $2k^{2}$ $(4k-1)$-packings that form a partition of $X_{2}$;
	\item two $(3\times 2k-1)$-packings that form a partition of a $(4k-1)$-packings from Point 2;
	\item $m^{2}$ $(3mk-1)$-packings that form a partition of a $(3k-1)$-packing from Point 1;
	\item $m^{2}$ $(4mk-1)$-packings that form a partition of a $(4k-1)$-packing from Point 2.
   \end{enumerate}
\label{sub1}
\end{prop}
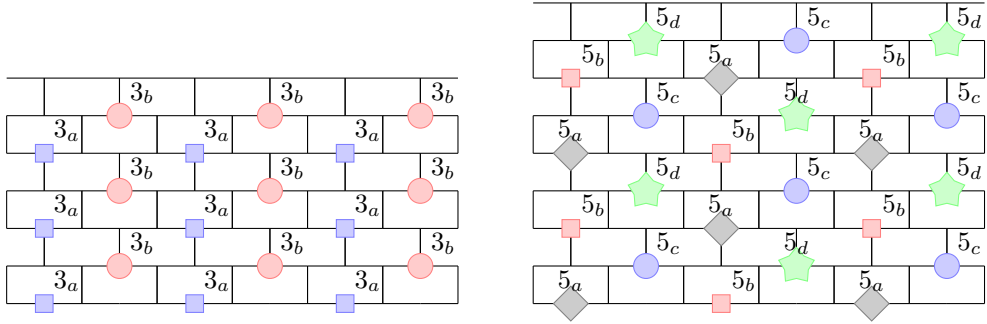
\begin{figure}[t]
\begin{center}
\begin{tikzpicture}
\foreach \x in {0,1,...,5}
\foreach \y in {0,1,2}
{
\draw (\x,\y) -- (\x-0.5,\y);
\draw (\x-0.5,\y) --(\x-0.5,\y+0.5);
\draw (\x-0.5,\y+0.5) -- (\x,\y+0.5);
\draw (\x,\y+0.5) -- (\x,\y+1);
\draw (\x,\y) -- (\x+0.5,\y);
\draw (\x+0.5,\y) -- (\x+0.5,\y+0.5);
\draw (\x+0.5,\y+0.5) -- (\x,\y+0.5);
\draw (\x,\y+0.5) -- (\x,\y+1);
}
\draw (-0.5,3) --(5.5,3);
\foreach \u in {0,2,4}
\foreach \v in {0,1,2}
{
\node at (\u,\v) [rectangle,draw=blue!50,fill=blue!20] {};
\node at (\u+0.3,\v+0.3){$3_a$};
\node at (\u+1,\v+0.5) [circle,draw=red!50,fill=red!20] {};
\node at (\u+1.3,\v+0.8){$3_b$};
}
\foreach \x in {7,8,...,12}
\foreach \y in {0,1,2,3}
{
\draw (\x,\y) -- (\x-0.5,\y);
\draw (\x-0.5,\y) --(\x-0.5,\y+0.5);
\draw (\x-0.5,\y+0.5) -- (\x,\y+0.5);
\draw (\x,\y+0.5) -- (\x,\y+1);
\draw (\x,\y) -- (\x+0.5,\y);
\draw (\x+0.5,\y) -- (\x+0.5,\y+0.5);
\draw (\x+0.5,\y+0.5) -- (\x,\y+0.5);
\draw (\x,\y+0.5) -- (\x,\y+1);
}
\draw (6.5,4) --(12.5,4);
\foreach \u in {7,11}
\foreach \v in {0,2}
{
\node at (\u,\v) [diamond,draw=black!50,fill=black!20] {};
\node at (\u,\v+0.3){$5_a$};
\node at (\u,\v+1) [rectangle,draw=red!50,fill=red!20] {};
\node at (\u+0.3,\v+1.3){$5_b$};
\node at (\u+1,\v+1.5) [star,draw=green!50,fill=green!20] {};
\node at (\u+1.3,\v+1.8){$5_d$};
\node at (\u+1,\v+0.5) [circle,draw=blue!50,fill=blue!20] {};
\node at (\u+1.3,\v+0.8){$5_c$};
}
\foreach \w in {9}
\foreach \s in {1,3}
{
\node at (\w,\s) [diamond,draw=black!50,fill=black!20] {};
\node at (\w,\s+0.3){$5_a$};
\node at (\w,\s-1) [rectangle,draw=red!50,fill=red!20] {};
\node at (\w+0.3,\s-0.7){$5_b$};
\node at (\w+1,\s-0.5) [star,draw=green!50,fill=green!20] {};
\node at (\w+1,\s-0.2){$5_d$};
\node at (\w+1,\s+0.5) [circle,draw=blue!50,fill=blue!20] {};
\node at (\w+1.3,\s+0.8){$5_c$};
}
\end{tikzpicture}
\end{center}
\caption{Two 3-packings forming a partition of $X_{2}$ (on the left) and four 5-packings forming a partition of $X_{2}$ (on the right).}
\end{figure}
\begin{proof}
1. Let $A_{k}$ be the $(3k-1)$-packing defined by $A_{k}=\{(2kx+4ky,kx)|\ x\in \mathbb{Z},\ y\in \mathbb{Z} \}$. Let $\mathcal{F}=\{(2i+4j,i)|i,j\in\{0,\ldots,k-1\}$ be a family of $k^{2}$ vectors. Make $k^{2}$ copies of the set $A_{k}$ and translate each one by a vector from $\mathcal{F}$ to obtain a partition of $X_{2}$.\newline
2. Let $B_{k}$ be the $(4k-1)$-packing defined by $B_{k}=\{(4kx,2ky)|\ x\in \mathbb{Z},\ y\in \mathbb{Z} \}$. Let $\mathcal{F}=\{(4i+2a,2j+a)|i,j\in\{0,\ldots,k-1\},\ a\in\{0,1\}\}$ be a family of $2k^{2}$ vectors. Make $2k^{2}$ copies of the set $B_{k}$ and translate each one by a vector from $\mathcal{F}$ to obtain a partition of $X_{2}$.\newline
3. Note that $A_{2k}\subseteq B_{k}$ and if $A'_{2k}$ is $A_{2k}$ translated by the vector $(0,2k)$, then $A'_{2k}\cup A_{2k}=B_{k}$.\newline
4. Note that $A_{mk}\subseteq A_{k}$. Let $\mathcal{F}=\{(2mki+4mkj,mki)|i,j\in\{0,\ldots,m-1\}\}$ be a family of $m^{2}$ vectors. Make $m^{2}$ copies of the set $A_{mk}$ and translate each one by a vector from $\mathcal{F}$ to obtain a partition of $A_{k}$.\newline
5. Note that $B_{mk}\subseteq B_{k}$. Let $\mathcal{F}=\{(4mki,2mkj)|i,j\in\{0,\ldots,m-1\}\}$ be a family of $m^{2}$ vectors. Make $m^{2}$ copies of the set $B_{mk}$ and translate each one by a vector from $\mathcal{F}$ to obtain a partition of $B_{k}$.
\end{proof}
Figure 3 illustrates a partition of $X_{2}$ from Points 1 and 2 for $k=1$.
In the remaining of the section, the proofs of decomposition of a set $X$ will be resumed in a table and the proofs of properties similar from those from Points 3, 4 and 5 will be left to the reader.
\subsection{Subdivision of a 3-packing in $\mathscr{H}$}
Let $X_{3}=\{(3x+6y,x)|\ x\in \mathbb{Z},\ y\in \mathbb{Z}\}$ be the maximized 3-packing in $\mathscr{H}$ from Figure 2. Note that $d(X_{3})=1/A(3)=1/6$ and that six 3-packings form a partition of $\mathscr{H}$ if we translate $X_{3}$ five times.
\begin{prop}
Let $k>0$ and $m>0$ be integers. There exist:

   \begin{enumerate}
\setlength{\parskip}{.05cm}
	\item $k^{2}$ $p_{1,k}$-packings, $p_{1,k}=(4k-1)$, that form a partition of $X_{3}$;
	\item $3k^{2}$ $p_{2,k}$-packings, $p_{2,k}=(6k-1)$, that form a partition of $X_{3}$;
	\item $8k^{2}$ $p_{3,k}$-packings, $p_{3,k}=(10k-1)$,  that form a partition of $X_{3}$;
	\item $24k^{2}$ $p_{4,k}$-packings, $p_{4,k}=(18k-1)$, that form a partition of $X_{3}$;
	\item $m^{2}$ $p_{j,mk}$-packings that form a partition of a $p_{j,k}$-packing from Point $j$, for $j\in\{1,\ldots,4\}$;
	\item three $(4\times 3k-1)$-packings that form a partition of a $(6k-1)$-packing from Point 2;
	\item two $(4\times 4k-1)$-packings that form a partition of a $(10k-1)$-packing from Point 3;
	\item four 17-packings and six 23-packings that form a partition of every 5-packing from Point 2.
   \end{enumerate}
\label{sub2}
\end{prop}
\begin{proof} The proof is resumed in Table~\ref{tab4}, this table contains: in which $i$-packing $X$ will be decomposed (Column 1), the number of $i$-packings needed to form a partition of $X$ (Column 2), the description of an $i$-packing with Cartesian coordinates (assuming $x$ and $y$ are integers) (Column 3) and the family of translation vectors (Column 4). We assume that if we do copies of this $i$-packing and we translate each one by one of these vectors. Afterward, we obtain a partition of $X$ in $i$-packings.
\end{proof}
\subsection{Subdivision of a 4-packing in $\mathscr{H}$}
Let $X_{4}=\{(3x+7y,2x+y)|\ x\in \mathbb{Z},\ y\in \mathbb{Z}\}$ be the 4-packing in $\mathscr{H}$ from Figure 2. Note that $d(X_{4})=1/11$ and that $1/A(4)=1/10$. However, we claim that a 4-packing with density 1/10 does not exist.
Note that eleven 4-packings form a partition of $\mathscr{H}$ if we translate $X_{4}$ ten times.

\begin{prop}
Let $k>0$ and $m>0$ be integers. There exist:

   \begin{enumerate}
\setlength{\parskip}{.05cm}
	\item $k^{2}$ $p_{1,k}$-packings, $p_{1,k}=(5k-1)$, that form a partition of $X_{4}$;
	\item $2k^{2}$ $p_{2,k}$-packings, $p_{2,k}=(6k-1)$, that form a partition of $X_{4}$;
	\item $3k^{2}$ $p_{3,k}$-packings, $p_{3,k}=(8k-1)$, that form a partition of $X_{4}$;
	\item $6k^{2}$ $p_{4,k}$-packings, $p_{4,k}=(11k-1)$, that form a partition of $X_{4}$;
	\item $m^{2}$ $p_{j,mk}$-packings that form a partition of a $p_{j,k}$-packing from Point $j$, for $j\in\{1,\ldots,4\}$;
	\item two $(5\times 2k-1)$-packings that form a partition of a $(6k-1)$-packing from Point 2;
	\item two $(6k-1)$-packings that form a partition of a $(5k-1)$-packing from Point 1;
	\item three $(5\times 3k-1)$-packings that form a partition of a $(8k-1)$-packing from Point 2;
	\item three $(8k-1)$-packings that form a partition of a $(5k-1)$-packing from Point 1.
   \end{enumerate}
\label{sub3}
\end{prop}
\begin{proof}
See Table~\ref{tab5}.
\end{proof}
\subsection{Subdivision of a 2-packing in $\mathbb{Z}^{2}$}
In the square lattice, we can choose one vertex as the origin and all the other vertices will be nominated by Cartesian coordinates. In all our representations our origin $(0,0)$ will be in the packing that we want to describe.
Let $X_{2}=\{(2x+y,x+3y)|\ x\in \mathbb{Z},\ y\in \mathbb{Z}\}$ be the maximized 2-packing in $\mathbb{Z}^{2}$ from Figure 4. Note that $d(X_{2})=1/A(2)=1/5$ and that five 2-packings form a partition of $\mathbb{Z}^{2}$ if we translate $X_{2}$ four times.\newline
\begin{figure}[t]
\begin{center}
\begin{tikzpicture}
\foreach \x in {-0.5,0,0.5,...,2.5}
{
\draw (\x,0) -- (\x,2);
}
\foreach \y in {0,0.5, ...,2}
{
\draw (-0.5,\y) -- (2.5,\y);
}
\foreach \u in {0,1}
\foreach \v in {0,1,2}
{
\node at (\v-\u/2,\v/2+\u) [circle,draw=blue!50,fill=blue!20] {};
\node at (\v-\u/2+0.3,\v/2+\u+0.3){2};
}
\node at (2.5,0) [circle,draw=blue!50,fill=blue!20] {};
\node at (2.8,0.3){2};
\foreach \x in {3,3.5,...,6}
{
\draw (\x,0) -- (\x,2);
}
\foreach \y in {0,0.5, ...,2}
{
\draw (3,\y) -- (6,\y);
}
\foreach \u in {3.5,5.5}
\foreach \v in {0,2}
{
\node at (\u,\v) [circle,draw=red!50,fill=red!20] {};
\node at (\u+0.3,\v+0.3){3};
}
\node at (4.5,1) [circle,draw=red!50,fill=red!20] {};
\node at (4.8,1.3){3};
\foreach \x in {6.5,7,7.5,...,12.5}
{
\draw (\x,0) -- (\x,2);
}
\foreach \y in {0,0.5, ...,2}
{
\draw (6.5,\y) -- (12.5,\y);
}
\foreach \v in {0,1,2}
{
\node at (7+3*\v/2,\v) [circle,draw=green!50,fill=green!20] {};
\node at (7+3*\v/2+0.3,\v+0.3){4};
}
\node at (11,0.5) [circle,draw=green!50,fill=green!20] {};
\node at (11.3,0.8){4};
\node at (12.5,1.5) [circle,draw=green!50,fill=green!20] {};
\node at (12.8,1.8){4};
\end{tikzpicture}
\end{center}
\caption{The sets $X_{2}$ (2-packing), $X_{3}$ (3-packing) and $X_{4}$ (4-packing) in $\mathbb{Z}^{2}$.}
\end{figure}
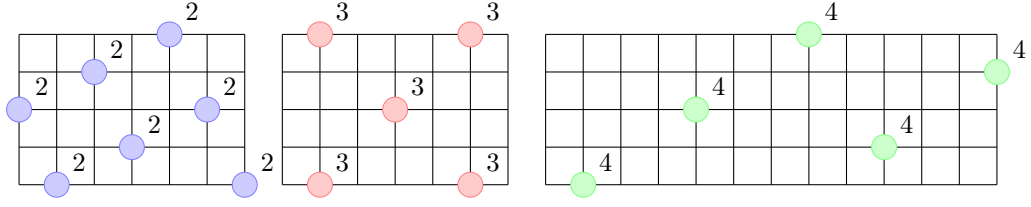
\begin{prop}
Let $k>0$, $m>0$ be integers. There exist:

 \begin{enumerate}
\setlength{\parskip}{.05cm}
	\item $k^{2}$ $(3k-1)$-packings that form a partition of $X_{2}$;
	\item $2k^{2}$ $(4k-1)$-packings that form a partition of $X_{2}$;
	\item two $(3\times 2k-1)$-packings that form a partition of a $(4k-1)$-packing from Point 2;
	\item two $(4k-1)$-packings that form a partition of a $(3k-1)$-packing from Point 1;
	\item $m^{2}$ $(3mk-1)$-packings that form a partition of a $(3k-1)$-packing from Point 1;
	\item $m^{2}$ $(4mk-1)$-packings that form a partition of a $(4k-1)$-packing from Point 2.
\end{enumerate}
\label{sub4}
\end{prop}
\begin{proof} 
See Table~\ref{tab7}.
\end{proof}
\subsection{Subdivision of a 3-packing in $\mathbb{Z}^{2}$}
Let $X_{3}=\{(2x+4y,2y)|\ x\in \mathbb{Z},\ y\in \mathbb{Z}\}$ be the maximized 3-packing in $\mathbb{Z}^{2}$ from Figure 4. Note that $d(X_{3})=1/A(3)=1/8$ and that eight 3-packings form a partition of $\mathbb{Z}^{2}$ if we translate $X_{3}$ seven times.

\begin{prop}
Let $k>0$ and $m>0$ be integers. There exist:

   \begin{enumerate}
\setlength{\parskip}{.05cm}
	\item $k^{2}$ $(4k-1)$-packings that form a partition of $X_{3}$;
	\item $m^{2}$ $(4mk-1)$-packings that form a partition of a $(4k-1)$-packing from Point 1.
   \end{enumerate}
\label{sub5}
\end{prop}
\begin{proof}
See Table~\ref{tab7}.
\end{proof}
\subsection{Subdivision of a 4-packing in $\mathbb{Z}^{2}$}
Let $X_{4}=\{(3x+8y,2x+y)|\ x\in \mathbb{Z},\ y\in \mathbb{Z}\}$ be the maximized 4-packing in $\mathbb{Z}^{2}$ from Figure 4. Note that $d(X_{4})=1/A(4)=1/13$ and that thirteen 4-packings form a partition of $\mathbb{Z}^{2}$ if we translate $X_{4}$ twelve times.
\begin{prop}
Let $k>0$, $m>0$ be integers. There exist:

   \begin{enumerate}
\setlength{\parskip}{.05cm}
	\item $k^{2}$ $(5k-1)$-packings that form a partition of $X_{4}$;
	\item $2k^{2}$ $(6k-1)$-packings that form a partition of $X_{4}$;
	\item two $(5\times 2k-1)$-packings that form a partition of a $(6k-1)$-packing from Point 2;
	\item two $(6k-1)$-packings that form a partition of a $(5k-1)$-packing from Point 1;
	\item $m^{2}$ $(5mk-1)$-packings that form a partition of a $(5k-1)$-packing from Point 1;
	\item $m^{2}$ $(6mk-1)$-packings that form a partition of a $(6k-1)$-packing from Point 2.
\end{enumerate}
\label{sub6}
\end{prop}
\begin{proof}
See Table~\ref{tab7}.
\end{proof}
\subsection{Subdivision of an independent set in $\mathscr{T}$}
The square lattice can be seen as a subgraph of the triangular lattice. In fact in Figure 5, the triangular lattice is represented as a supergraph of the square lattice. Therefore, we can choose one vertex as the origin and all the other vertices will be nominated by Cartesian coordinates. In all our representations our origin $(0,0)$ will be a vertex such that $(0,0)$ has $(1,0)$, $(0,1)$, $(-1,0)$, $(0,-1)$, $(-1,1)$ and $(1,-1)$ as neighbors.
Let $X_{1}=\{(x+3y,x)|\ x\in \mathbb{Z},\ y\in \mathbb{Z} \}$ be the (unique) maximized independent set (1-packing) in $\mathscr{T}$ from Figure 5. Note that $d(X_{1})=1/A(1)=1/3$ and that three independent sets form a partition of $\mathscr{T}$ if we translate $X_{1}$ two times.
\begin{figure}[t]
\begin{center}
\begin{tikzpicture}
\foreach \x in {-0.5,0,0.5,...,1.5}
{
\draw (\x,0) -- (\x,2);
}
\foreach \y in {0,0.5, ...,2}
{
\draw (-0.5,\y) -- (1.5,\y);
}
\foreach \z in {-0.5}
{
\draw (\z,2) -- (\z+2,0);
}
\foreach \h in {0.5,1,1.5}
{
\draw (-0.5,\h) -- (\h-0.5,0);
\draw (\h-0.5,2) -- (1.5,\h);
}
\foreach \u in {0,0.5,1}
\foreach \v in {0,1}
{
\node at (\u-\v/2,\u+\v) [circle,draw=blue!50,fill=blue!20] {};
\node at (\u-\v/2+0.3,\u+\v+0.1){1};
}
\node at (1.5,1.5) [circle,draw=blue!50,fill=blue!20] {};
\node at (1.8,1.6){1};
\node at (1.5,0) [circle,draw=blue!50,fill=blue!20] {};
\node at (1.8,0.1){1};
\foreach \x in {2,2.5,3,...,6}
{
\draw (\x,0) -- (\x,2);
}
\foreach \y in {0,0.5, ...,2}
{
\draw (2,\y) -- (6,\y);
}
\foreach \z in {2,2.5,3,3.5,4}
{
\draw (\z,2) -- (\z+2,0);
}
\foreach \h in {0.5,1,1.5}
{
\draw (2,\h) -- (\h+2,0);
\draw (\h+4,2) -- (6,\h);
}
\foreach \u in {0,0.5,1,1.5}
{
\node at (2*\u+2.5,\u) [circle,draw=red!50,fill=red!20] {};
\node at (2*\u+2.8,\u+0.1){2};
}
\foreach \u in {0,0.5}
{
\node at (2*\u+2,\u+1.5) [circle,draw=red!50,fill=red!20] {};
\node at (2*\u+2.3,\u+1.6){2};
}
\node at (6,0.1) [circle,draw=red!50,fill=red!20] {};
\node at (6.3,0.1){2};
\foreach \x in {6.5,7,7.5,...,12.5}
{
\draw (\x,0) -- (\x,2);
}
\foreach \y in {0,0.5, ...,2}
{
\draw (6.5,\y) -- (12.5,\y);
}
\foreach \z in {6.5,7,7.5, ...,10.5}
{
\draw (\z,2) -- (\z+2,0);
}
\foreach \h in {0.5,1,1.5}
{
\draw (6.5,\h) -- (\h+6.5,0);
\draw (\h+10.5,2) -- (12.5,\h);
}
\foreach \u in {0,1,2}
{
\node at (\u+7,\u) [circle,draw=green!50,fill=green!20] {};
\node at (\u+7.3,\u+0.1){3};
\node at (\u+10,\u) [circle,draw=green!50,fill=green!20] {};
\node at (\u+10.3,\u+0.1){3};
}
\end{tikzpicture}
\end{center}
\caption{The sets $X_{1}$ (1-packing), $X_{2}$ (2-packing) and $X_{3}$ (3-packing) in $\mathscr{T}$.}
\end{figure}
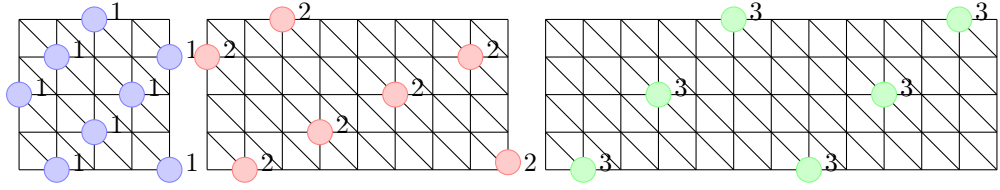
\begin{prop}
Let $k>0$ and $m>0$ be integers. There exist:
\begin{enumerate}
\setlength{\parskip}{.05cm}
	\item $k^{2}$ $(2k-1)$-packings that form a partition of $X_{1}$;
	\item $3k^{2}$ $(3k-1)$-packings that form a partition of $X_{1}$;
	\item three $(3k-1)$-packings that form a partition of a $(2k-1)$-packing from Point 1;
	\item $m^{2}$ $(2mk-1)$-packings that form a partition of a $(2k-1)$-packing from Point 1;
	\item $m^{2}$ $(3mk-1)$-packings that form a partition of a $(3k-1)$-packing from Point 2.
\end{enumerate}
\label{sub7}
\end{prop}
\begin{proof}
See Table~\ref{tab10}.
\end{proof}
\subsection{Subdivision of a 2-packing in $\mathscr{T}$}
Let $X_{2}=\{(2x-y,x+3y)|\ x\in \mathbb{Z},\ y\in \mathbb{Z}\}$ be the maximized 2-packing in $\mathscr{T}$ from Figure 5. Note that $d(X_{2})=1/A(2)=1/7$ and that seven 2-packings form a partition of $\mathscr{T}$ if we translate $X_{2}$ six times.
\begin{prop}
Let $k>0$ and $m>0$ be integers. There exist:
\begin{enumerate}
\setlength{\parskip}{.05cm}
	\item $k^{2}$ $(3k-1)$-packings that form a partition of $X_{2}$;
	\item $m^{2}$ $(3mk-1)$-packings that form a partition of a $(3k-1)$-packing from Point 1.
\end{enumerate}
\label{sub8}
\end{prop}
\begin{proof}
See Table~\ref{tab10}.
\end{proof}
\subsection{Subdivision of a 3-packing in $\mathscr{T}$}
Let $X_{3}=\{(2x+6y,2x)|\ x\in \mathbb{Z},\ y\in \mathbb{Z}\}$ be the maximized 3-packing in $\mathscr{T}$ from Figure 5. Note that $d(X_{3})=1/A(3)=1/12$ and that twelve 3-packings form a partition of $\mathscr{T}$ if we translate $X_{3}$ eleven times.
\begin{prop}
Let $k>0$ and $m>0$ be integers. There exist:
\begin{enumerate}
\setlength{\parskip}{.05cm}
	\item $k^{2}$ $(4k-1)$-packings that form a partition of $X_{3}$;
	\item $3k^{2}$ $(6k-1)$-packings that form a partition of $X_{3}$;
	\item three $(4\times 3k-1)$-packings that form a partition of a $(6k-1)$-packing from Point 1;
	\item $m^{2}$ $(4mk-1)$-packings that form a partition of a $(4k-1)$-packing from Point 1;
	\item $m^{2}$ $(6mk-1)$-packings that form a partition of a $(6k-1)$-packing from Point 2.
\end{enumerate}
\label{sub9}
\end{prop}
\begin{proof}
See Table~\ref{tab10}.
\end{proof}
\section{$S$-packing chromatic number}
\subsection{General properties}
In the previous section, we obtained several properties of subdivision of an $i$-packings in a lattice.
This section illustrates general properties obtained on the $S$-packing chromatic number using only a small part of these properties. 
For a given sequence $S$, one can find other colorings of a lattice using properties from the previous section.
\begin{cor}
Let $a_{0}=1$.
If $s_{1}=2$ and there exist three integers $1<a_{1}<\ldots<a_{3}$ and three integers $k_{1},\ldots,k_{3}$  such that $s_{a_{i}}\le 3k_{i}-1$ and $a_{i}-a_{i-1}\ge k_{i}^{2}$ or $s_{a_{i}}\le 4k_{i}-1$ and $a_{i}-a_{i-1}\ge 2k_{i}^{2}$ for $i\in\{1,\ldots,3\}$ then $\chi^{S}_{\rho}(\mathscr{H})\le a_{3}$.
\end{cor}
This corollary can be useful to find upper bounds for a given sequence.
For example, if $S=(2,2,2,2,\ldots)$, then taking $a_{1}=2$, $a_{2}=3$ and $a_{3}=4$, Corollary 4.1 gives us $\chi^{S}_{\rho}(\mathscr{H})\le 4$ (this result is in fact treated in next subsection).
Similarly, for the sequence $S=(2,3,3,5,5,5,5,7,7,7,7,7,7,7,7,\ldots)$ , then taking $a_{1}=3$, $a_{2}=7$ and $a_{3}=15$, Corollary 4.1 gives us $\chi^{S}_{\rho}(\mathscr{H})\le 15$.
There are similar results for $s_{1}=3$ or $s_{1}=4$ using Propositions 3.3 and 3.4.\newline
For the two remaining lattices, the two following properties are given for $\mathbb{Z}^{2}$ with $s_1=2$ and for $\mathscr{T}$ with $s_1=1$. There exist similar properties for $\mathbb{Z}^{2}$ with $s_1=3$ or 4 using Propositions 3.5 and 3.6 and for $\mathscr{T}$ with $s_1=2$ or 3 using Propositions 3.8 and 3.9.
\begin{cor}
Let $a_{0}=1$.
If $s_{1}=2$ and there exist four integers $1<a_{1}<\ldots<a_{4}$ and four integers $k_{1},\ldots,k_{4}$  such that $s_{a_{i}}\le 3k_{i}-1$ and $a_{i}-a_{i-1}\ge k_{i}^{2}$ or $s_{a_{i}}\le 4k_{i}-1$ and $a_{i}-a_{i-1}\ge 2k_{i}^{2}$ for $i\in\{1,\ldots,4\}$ then $\chi^{S}_{\rho}(\mathbb{Z}^{2})\le a_{4}$.
\end{cor}
\begin{cor}
Let $a_{0}=1$.
If $s_{1}=1$ and there exist two integers $1<a_{1}<a_{2}$ and two integers $k_{1}$ and $k_{2}$  such that $s_{a_{i}}\le 2k_{i}-1$ and $a_{i}-a_{i-1}\ge k_{i}^{2}$ or $s_{a_{i}}\le 3k_{i}-1$ and $a_{i}-a_{i-1}\ge 3k_{i}^{2}$ for $i\in\{1,\ldots,2\}$ then $\chi^{S}_{\rho}(\mathscr{T})\le a_{2}$.
\end{cor}
\subsection{$S$-packing chromatic number and distance coloring}
Jacko and Jendrol \cite{JA2005}, Fertin \textit{et al.} \cite{FE2003} and Ševcıková \cite{SE2001} have studied distance colorings of $\mathscr{H}$, $\mathbb{Z}^{2}$ and $\mathscr{T}$ respectively. The following propositions comes from their work and can be translated in $S$-packing coloring:
\begin{prop}[\cite{JA2005}]
Let $n$ and $d$ be integers. The minimum $n$ such that $s_{1}=d$, $s_{n}=d$ and $\chi^{S}_{\rho}(\mathscr{H})=n$ is given by
$$n=\left\{
    \begin{array}{ll}
	\lceil\frac{3}{8}(d+1)^{2}\rceil & \mbox{for $d$ odd;}\\
	\lceil\frac{3}{8}(d+4/3)^{2}\rceil & \mbox{for $d$ even}.\\
\end{array}
\right.$$
\end{prop}
\begin{prop}[\cite{FE2003}]
Let $n$ and $d$ be integers. The minimum $n$ such that $s_{1}=d$, $s_{n}=d$ and $\chi^{S}_{\rho}(\mathbb{Z}^{2})=n$ is given by
$$n=\left\{
    \begin{array}{ll}
	\frac{1}{2}(d+1)^{2} & \mbox{for $d$ odd;}\\
	\frac{1}{2}((d+1)^{2}+1) & \mbox{for $d$ even}.\\
\end{array}
\right.$$
\end{prop}
\begin{prop}[\cite{SE2001}]
Let $n$ and $d$ be integers. The minimum $n$ such that $s_{1}=d$, $s_{n}=d$ and $\chi^{S}_{\rho}(\mathscr{T})=n$ is given by
$$n=\left\lceil\frac{3}{4}(d+1)^{2}\right\rceil.$$
\end{prop}
\section{$(d,n)$-packing chromatic number}
\subsection{Hexagonal lattice}
\begin{prop}
$\chi^{2,1}_{\rho}(\mathscr{H})=\infty, \ \chi^{5,2}_{\rho}(\mathscr{H})=\infty, \ \chi^{8,3}_{\rho}(\mathscr{H})=\infty, \ \chi^{11,4}_{\rho}(\mathscr{H})=\infty, \ \chi^{13,5}_{\rho}(\mathscr{H})=\infty\ and\ \chi^{16,6}_{\rho}(\mathscr{H})=\infty$.
\end{prop}
\begin{proof}
Let $\mathscr{H}$ be the hexagonal lattice and $k$ be an integer, $k\ge16$.\newline 
$\sum\limits_{i=1}^{k}\frac{1}{A(i)}=\sum\limits_{i=1}^{n}\frac{1}{A(2i)}+\sum\limits_{i=0}^{n}\frac{1}{A(4i+1)}+\sum\limits_{i=0}^{k}\frac{1}{A(4i+3)}=
\sum\limits_{i=1}^{k}\frac{1}{\frac{3}{2}i^{2}+\frac{3}{2}i +1}+\sum\limits_{i=0}^{k}\frac{1}{6i^{2}+6i +2}+\sum\limits_{i=0}^{k}\frac{1}{6i^{2}+12i +6}<\frac{2}{15}\sqrt{15}\pi\tanh(\frac{1}{6}\pi\sqrt{15})+\frac{1}{6}\sqrt{3}\pi\tanh(\frac{1}{6}\pi\sqrt{3})+\frac{1}{36}\pi^{2}-1<1.494.$\newline
Therefore: $\sum\limits_{i=2}^{k}\frac{1}{A(i)}<1.494-\frac{1}{A(1)}<0.994<1$, 
$\sum\limits_{i=5}^{k}\frac{2}{A(i)}<2(1.494-\sum\limits_{i=1}^{4}\frac{1}{A(i)})<0.955<1$, 
$\sum\limits_{i=8}^{k}\frac{3}{A(i)}<3(1.494-\sum\limits_{i=1}^{7}\frac{1}{A(i)})<0.935<1$,
$\sum\limits_{i=11}^{k}\frac{4}{A(i)}<4(1.494-\sum\limits_{i=1}^{10}\frac{1}{A(i)})<0.925<1$ and 
$\sum\limits_{i=13}^{k}\frac{5}{A(i)}<5(1.494-\sum\limits_{i=1}^{12}\frac{1}{A(i)})<0.986<1$, $\sum\limits_{i=16}^{k}\frac{6}{A(i)}<6(1.494-\sum\limits_{i=1}^{15}\frac{1}{A(i)})<0.968<1$.
Corollary 2.3 allows us to conclude.
\end{proof}
\begin{prop}
$\chi^{2,2}_{\rho}(\mathscr{H})\le 8$, $\chi^{2,3}_{\rho}(\mathscr{H})\le 5$ and $\forall n\ge 4$, $\chi^{2,n}_{\rho}(\mathscr{H})=4$.
\end{prop}
\begin{proof}
Using Proposition~\ref{sub1}, we define a $(2,n)$-packing coloring of $\mathscr{H}$ for each $n=2,3$ and $n\ge4$.
$\mathscr{H}$ can be partitioned into four 2-packings, the first two ones can be colored by color 2, the third one by two colors 3 and the last one by four colors under 5, it will be two 4 and two 5, to conclude $\chi^{2,2}_{\rho}(\mathscr{H})\le8$.
$\mathscr{H}$ can be partitioned into four 2-packings, the first three ones can be colored by colors 2 and the third one by two colors 3, to conclude $\chi^{2,3}_{\rho}(\mathscr{H})\le5$.
$\mathscr{H}$ can be partitioned into four 2-packings, hence $\forall n\ge 4$, $\chi^{2,n}_{\rho}(\mathscr{H})=4$.
\end{proof}
The following table summarizes the colorings defined in the above proof. The symbol P in the table refers to the packings we use and how we subdivide them into $i$-packings ($A_{i}$ is an $i$-packing) and the symbol C refers to the associated colors we use for each $i$-packing.
By $k\times A_{i}$ we mean we use $k$ $i$-packings, and  by $k\times i$ we mean we use $k$ colors $i$.
In the rest of the paper, similar proofs will be only described by a table using the same format than this one.
\begin{center}
\begin{tabular}{| c | c | c | c | c |}
\hline
\multirow{3}{*}{(2,2)-packing} & \multirow{2}{*}{P} & \multirow{2}{*}{2$\times X_{2}$} & $X_{2}$ & $X_{2}$ \\ \cline{4-5}
& & & 2$\times A_{3}$ & 4$\times A_{5}$ \\ \cline{2-5}
& C & 2$\times$2 & 2$\times$3 & 2$\times$4, 2$\times$5 \\ \hline
\multirow{3}{*}{(2,3)-packing} & \multirow{2}{*}{P} & \multirow{2}{*}{3$\times X_{2}$} & $X_{2}$ & \\ \cline{4-5}
& & & 2$\times A_{3}$ & \\ \cline{2-5}
& C & 3$\times$2 & 2$\times$3 & \\ \hline
\end{tabular}
\end{center}
\begin{prop}
$\chi^{3,2}_{\rho}(\mathscr{H})\le 35$, $\chi^{3,3}_{\rho}(\mathscr{H})\le 13$, $\chi^{3,4}_{\rho}(\mathscr{H})\le 10$, $\chi^{3,5}_{\rho}(\mathscr{H})\le 8$ and $\forall n\ge 6$, $\chi^{3,n}_{\rho}(\mathscr{H})=6$.
\end{prop}
\begin{proof}
Using Proposition~\ref{sub2}, we define a $(3,n)$-packing coloring of $\mathscr{H}$ for each $n=2,3,4,5$ and $n\ge6$.
$\mathscr{H}$ can be partitioned into six 3-packings, hence $\forall n\ge 6$, $\chi^{3,n}_{\rho}(\mathscr{H})=6$. The other colorings are described in Table~\ref{tabc1}.
\end{proof}
\begin{prop}
$\chi^{4,3}_{\rho}(\mathscr{H})\le 58$, $\chi^{4,4}_{\rho}(\mathscr{H})\le 27$, $\chi^{4,5}_{\rho}(\mathscr{H})\le 21$, $\chi^{4,6}_{\rho}(\mathscr{H})\le 18$ and from \cite{JA2005} $\forall n\ge 11$, $\chi^{4,n}_{\rho}(\mathscr{H})= 11$.
\end{prop}
\begin{proof}
Using Proposition~\ref{sub3}, we define a $(4,n)$-packing coloring of $\mathscr{H}$ for each $n=3,4,5,6$ and $n\ge11$.
$\mathscr{H}$ can be partitioned into eleven 4-packings. The other colorings are described in Table~\ref{tabc2}.
\end{proof}
\subsection{Square lattice}
\begin{prop}
$\chi^{2,1}_{\rho}(\mathbb{Z}^{2})=\infty, \ \chi^{4,2}_{\rho}(\mathbb{Z}^{2})=\infty, \ \chi^{6,3}_{\rho}(\mathbb{Z}^{2})=\infty, \ \chi^{8,4}_{\rho}(\mathbb{Z}^{2})=\infty, \ \chi^{10,5}_{\rho}(\mathbb{Z}^{2})=\infty\ and \ \chi^{12,6}_{\rho}(\mathbb{Z}^{2})=\infty$.
\end{prop}
\begin{proof}
Let $\mathbb{Z}^{2}$ be the square lattice and $k$ be an integer, $k\ge12$.\newline
$\sum\limits_{i=1}^{k}\frac{1}{A(i)}=\sum\limits_{i=1}^{k}\frac{1}{A(2i)}+\sum\limits_{i=0}^{k}\frac{1}{A(2i+1)}=\sum\limits_{i=1}^{k}\frac{1}{2i^{2}+2i+1}+\sum\limits_{i=0}^{k}\frac{1}{2i^{2}+4i +2}<\frac{1}{2}\pi \tanh(\frac{1}{2}\pi)+\frac{1}{12}\pi^{2}-1<1.264$.\newline
Therefore: $\sum\limits_{i=2}^{k}\frac{1}{A(i)}<1.264-\frac{1}{A(1)}<0.764<1$, 
$\sum\limits_{i=4}^{k}\frac{2}{A(i)}<2(1.264-\sum\limits_{i=1}^{3}\frac{1}{A(i)})<0.877<1$, 
$\sum\limits_{i=6}^{k}\frac{3}{A(i)}<3(1.264-\sum\limits_{i=1}^{5}\frac{1}{A(i)})<0.917<1$, 
$\sum\limits_{i=8}^{k}\frac{4}{A(i)}<4(1.264-\sum\limits_{i=1}^{7}\frac{1}{A(i)})<0.938<1$, 
$\sum\limits_{i=10}^{k}\frac{5}{A(i)}<5(1.264-\sum\limits_{i=1}^{9}\frac{1}{A(i)})<0.951<1$ and 
$\sum\limits_{i=12}^{k}\frac{6}{A(i)}<6(1.264-\sum\limits_{i=1}^{11}\frac{1}{A(i)})<0.959<1$.
Corollary 2.3 allows us to conclude.
\end{proof}
\begin{prop}
$\chi^{2,2}_{\rho}(\mathbb{Z}^{2})\le 20$, $\chi^{2,3}_{\rho}(\mathbb{Z}^{2})\le 8$, $\chi^{2,4}_{\rho}(\mathbb{Z}^{2})\le 6$ and $\forall n\ge 5$, $\chi^{2,n}_{\rho}(\mathbb{Z}^{2})= 5$.
\end{prop}
\begin{proof}
Using Proposition~\ref{sub4}, we define a $(2,n)$-packing coloring of $\mathbb{Z}^{2}$ for each $n=2,3,4$ and $n\ge5$.
$\mathbb{Z}^{2}$ can be partitioned into five 2-packings, hence $\forall n\ge 5$,  $\chi^{2,n}_{\rho}(\mathbb{Z}^{2})= 5$. The other colorings are described in Table~\ref{tabc4}.
\end{proof}
Soukal and Holub \cite{SO2010} have proven that $\chi^{1,1}_{\rho}(\mathbb{Z}^{2})\le 17$, and proposed a $24 \times 24$ pattern in order to color the square lattice.
Their pattern is recalled in Figure 6.
\begin{figure}[t]
\begin{center}
\begin{tikzpicture}
\foreach \u in {0,0.75, ...,8.25}
\foreach \v in {0,0.75, ...,8.25}
{
\node at (\u,\v){1};
\node at (\u+0.375,\v-0.375){1};
}
\node at (0,-0.25*1.5){3};
\node at (0.5*1.5,-0.25*1.5){4};
\node at (1*1.5,-0.25*1.5){8};
\node at (1.5*1.5,-0.25*1.5){5};
\node at (2*1.5,-0.25*1.5){3};
\node at (2.5*1.5,-0.25*1.5){2};
\node at (3*1.5,-0.25*1.5){3};
\node at (3.5*1.5,-0.25*1.5){6};
\node at (4*1.5,-0.25*1.5){12};
\node at (4.5*1.5,-0.25*1.5){7};
\node at (5*1.5,-0.25*1.5){3};
\node at (5.5*1.5,-0.25*1.5){2};
\node at (0,0.25*1.5){2};
\node at (0.5*1.5,0.25*1.5){3};
\node at (1*1.5,0.25*1.5){7};
\node at (1.5*1.5,0.25*1.5){4};
\node at (2*1.5,0.25*1.5){6};
\node at (2.5*1.5,0.25*1.5){3};
\node at (3*1.5,0.25*1.5){2};
\node at (3.5*1.5,0.25*1.5){3};
\node at (4*1.5,0.25*1.5){5};
\node at (4.5*1.5,0.25*1.5){4};
\node at (5*1.5,0.25*1.5){8};
\node at (5.5*1.5,0.25*1.5){3};
\node at (0,0.75*1.5){3};
\node at (0.5*1.5,0.75*1.5){2};
\node at (1*1.5,0.75*1.5){3};
\node at (1.5*1.5,0.75*1.5){12};
\node at (2*1.5,0.75*1.5){5};
\node at (2.5*1.5,0.75*1.5){4};
\node at (3*1.5,0.75*1.5){3};
\node at (3.5*1.5,0.75*1.5){2};
\node at (4*1.5,0.75*1.5){3};
\node at (4.5*1.5,0.75*1.5){6};
\node at (5*1.5,0.75*1.5){5};
\node at (5.5*1.5,0.75*1.5){7};
\node at (0,1.25*1.5){5};
\node at (0.5*1.5,1.25*1.5){3};
\node at (1*1.5,1.25*1.5){2};
\node at (1.5*1.5,1.25*1.5){3};
\node at (2*1.5,1.25*1.5){7};
\node at (2.5*1.5,1.25*1.5){8};
\node at (3*1.5,1.25*1.5){6};
\node at (3.5*1.5,1.25*1.5){3};
\node at (4*1.5,1.25*1.5){2};
\node at (4.5*1.5,1.25*1.5){3};
\node at (5*1.5,1.25*1.5){4};
\node at (5.5*1.5,1.25*1.5){15};
\node at (0,1.75*1.5){4};
\node at (0.5*1.5,1.75*1.5){7};
\node at (1*1.5,1.75*1.5){3};
\node at (1.5*1.5,1.75*1.5){2};
\node at (2*1.5,1.75*1.5){3};
\node at (2.5*1.5,1.75*1.5){5};
\node at (3*1.5,1.75*1.5){4};
\node at (3.5*1.5,1.75*1.5){9};
\node at (4*1.5,1.75*1.5){3};
\node at (4.5*1.5,1.75*1.5){2};
\node at (5*1.5,1.75*1.5){3};
\node at (5.5*1.5,1.75*1.5){6};
\node at (0,2.25*1.5){8};
\node at (0.5*1.5,2.25*1.5){5};
\node at (1*1.5,2.25*1.5){4};
\node at (1.5*1.5,2.25*1.5){3};
\node at (2*1.5,2.25*1.5){2};
\node at (2.5*1.5,2.25*1.5){3};
\node at (3*1.5,2.25*1.5){7};
\node at (3.5*1.5,2.25*1.5){5};
\node at (4*1.5,2.25*1.5){6};
\node at (4.5*1.5,2.25*1.5){3};
\node at (5*1.5,2.25*1.5){2};
\node at (5.5*1.5,2.25*1.5){3};
\node at (0,2.75*1.5){3};
\node at (0.5*1.5,2.75*1.5){6};
\node at (1*1.5,2.75*1.5){13};
\node at (1.5*1.5,2.75*1.5){7};
\node at (2*1.5,2.75*1.5){3};
\node at (2.5*1.5,2.75*1.5){2};
\node at (3*1.5,2.75*1.5){3};
\node at (3.5*1.5,2.75*1.5){4};
\node at (4*1.5,2.75*1.5){8};
\node at (4.5*1.5,2.75*1.5){5};
\node at (5*1.5,2.75*1.5){3};
\node at (5.5*1.5,2.75*1.5){2};
\node at (0,3.25*1.5){2};
\node at (0.5*1.5,3.25*1.5){3};
\node at (1*1.5,3.25*1.5){5};
\node at (1.5*1.5,3.25*1.5){4};
\node at (2*1.5,3.25*1.5){8};
\node at (2.5*1.5,3.25*1.5){3};
\node at (3*1.5,3.25*1.5){2};
\node at (3.5*1.5,3.25*1.5){3};
\node at (4*1.5,3.25*1.5){7};
\node at (4.5*1.5,3.25*1.5){4};
\node at (5*1.5,3.25*1.5){6};
\node at (5.5*1.5,3.25*1.5){3};
\node at (0,3.75*1.5){3};
\node at (0.5*1.5,3.75*1.5){2};
\node at (1*1.5,3.75*1.5){3};
\node at (1.5*1.5,3.75*1.5){6};
\node at (2*1.5,3.75*1.5){5};
\node at (2.5*1.5,3.75*1.5){7};
\node at (3*1.5,3.75*1.5){3};
\node at (3.5*1.5,3.75*1.5){2};
\node at (4*1.5,3.75*1.5){3};
\node at (4.5*1.5,3.75*1.5){9};
\node at (5*1.5,3.75*1.5){5};
\node at (5.5*1.5,3.75*1.5){4};
\node at (0,4.25*1.5){6};
\node at (0.5*1.5,4.25*1.5){3};
\node at (1*1.5,4.25*1.5){2};
\node at (1.5*1.5,4.25*1.5){3};
\node at (2*1.5,4.25*1.5){4};
\node at (2.5*1.5,4.25*1.5){14};
\node at (3*1.5,4.25*1.5){5};
\node at (3.5*1.5,4.25*1.5){3};
\node at (4*1.5,4.25*1.5){2};
\node at (4.5*1.5,4.25*1.5){3};
\node at (5*1.5,4.25*1.5){7};
\node at (5.5*1.5,4.25*1.5){8};
\node at (0,4.75*1.5){4};
\node at (0.5*1.5,4.75*1.5){9};
\node at (1*1.5,4.75*1.5){3};
\node at (1.5*1.5,4.75*1.5){2};
\node at (2*1.5,4.75*1.5){3};
\node at (2.5*1.5,4.75*1.5){6};
\node at (3*1.5,4.75*1.5){4};
\node at (3.5*1.5,4.75*1.5){7};
\node at (4*1.5,4.75*1.5){3};
\node at (4.5*1.5,4.75*1.5){2};
\node at (5*1.5,4.75*1.5){3};
\node at (5.5*1.5,4.75*1.5){5};
\node at (0,5.25*1.5){7};
\node at (0.5*1.5,5.25*1.5){5};
\node at (1*1.5,5.25*1.5){6};
\node at (1.5*1.5,5.25*1.5){3};
\node at (2*1.5,5.25*1.5){2};
\node at (2.5*1.5,5.25*1.5){3};
\node at (3*1.5,5.25*1.5){8};
\node at (3.5*1.5,5.25*1.5){5};
\node at (4*1.5,5.25*1.5){4};
\node at (4.5*1.5,5.25*1.5){3};
\node at (5*1.5,5.25*1.5){2};
\node at (5.5*1.5,5.25*1.5){3};
\node at (0.25*1.5,0){11};
\node at (0.75*1.5,0){2};
\node at (1.25*1.5,0){3};
\node at (1.75*1.5,0){2};
\node at (2.25*1.5,0){13};
\node at (2.75*1.5,0){5};
\node at (3.25*1.5,0){4};
\node at (3.75*1.5,0){2};
\node at (4.25*1.5,0){3};
\node at (4.75*1.5,0){2};
\node at (5.25*1.5,0){9};
\node at (5.75*1.5,0){5};
\node at (0.25*1.5,0.5*1.5){6};
\node at (0.75*1.5,0.5*1.5){5};
\node at (1.25*1.5,0.5*1.5){2};
\node at (1.75*1.5,0.5*1.5){3};
\node at (2.25*1.5,0.5*1.5){2};
\node at (2.75*1.5,0.5*1.5){11};
\node at (3.25*1.5,0.5*1.5){7};
\node at (3.75*1.5,0.5*1.5){10};
\node at (4.25*1.5,0.5*1.5){2};
\node at (4.75*1.5,0.5*1.5){3};
\node at (5.25*1.5,0.5*1.5){2};
\node at (5.75*1.5,0.5*1.5){4};
\node at (0.25*1.5,1*1.5){10};
\node at (0.75*1.5,1*1.5){4};
\node at (1.25*1.5,1*1.5){9};
\node at (1.75*1.5,1*1.5){2};
\node at (2.25*1.5,1*1.5){3};
\node at (2.75*1.5,1*1.5){2};
\node at (3.25*1.5,1*1.5){5};
\node at (3.75*1.5,1*1.5){4};
\node at (4.25*1.5,1*1.5){16};
\node at (4.75*1.5,1*1.5){2};
\node at (5.25*1.5,1*1.5){3};
\node at (5.75*1.5,1*1.5){2};
\node at (0.25*1.5,1.5*1.5){2};
\node at (0.75*1.5,1.5*1.5){17};
\node at (1.25*1.5,1.5*1.5){5};
\node at (1.75*1.5,1.5*1.5){4};
\node at (2.25*1.5,1.5*1.5){2};
\node at (2.75*1.5,1.5*1.5){3};
\node at (3.25*1.5,1.5*1.5){2};
\node at (3.75*1.5,1.5*1.5){14};
\node at (4.25*1.5,1.5*1.5){5};
\node at (4.75*1.5,1.5*1.5){11};
\node at (5.25*1.5,1.5*1.5){2};
\node at (5.75*1.5,1.5*1.5){3};
\node at (0.25*1.5,2*1.5){3};
\node at (0.75*1.5,2*1.5){2};
\node at (1.25*1.5,2*1.5){11};
\node at (1.75*1.5,2*1.5){6};
\node at (2.25*1.5,2*1.5){10};
\node at (2.75*1.5,2*1.5){2};
\node at (3.25*1.5,2*1.5){3};
\node at (3.75*1.5,2*1.5){2};
\node at (4.25*1.5,2*1.5){4};
\node at (4.75*1.5,2*1.5){7};
\node at (5.25*1.5,2*1.5){5};
\node at (5.75*1.5,2*1.5){2};
\node at (0.25*1.5,2.5*1.5){2};
\node at (0.75*1.5,2.5*1.5){3};
\node at (1.25*1.5,2.5*1.5){2};
\node at (1.75*1.5,2.5*1.5){5};
\node at (2.25*1.5,2.5*1.5){4};
\node at (2.75*1.5,2.5*1.5){15};
\node at (3.25*1.5,2.5*1.5){2};
\node at (3.75*1.5,2.5*1.5){3};
\node at (4.25*1.5,2.5*1.5){2};
\node at (4.75*1.5,2.5*1.5){10};
\node at (5.25*1.5,2.5*1.5){4};
\node at (5.75*1.5,2.5*1.5){9};
\node at (0.25*1.5,3*1.5){4};
\node at (0.75*1.5,3*1.5){2};
\node at (1.25*1.5,3*1.5){3};
\node at (1.75*1.5,3*1.5){2};
\node at (2.25*1.5,3*1.5){9};
\node at (2.75*1.5,3*1.5){5};
\node at (3.25*1.5,3*1.5){11};
\node at (3.75*1.5,3*1.5){2};
\node at (4.25*1.5,3*1.5){3};
\node at (4.75*1.5,3*1.5){2};
\node at (5.25*1.5,3*1.5){12};
\node at (5.75*1.5,3*1.5){5};
\node at (0.25*1.5,3.5*1.5){7};
\node at (0.75*1.5,3.5*1.5){10};
\node at (1.25*1.5,3.5*1.5){2};
\node at (1.75*1.5,3.5*1.5){3};
\node at (2.25*1.5,3.5*1.5){2};
\node at (2.75*1.5,3.5*1.5){4};
\node at (3.25*1.5,3.5*1.5){6};
\node at (3.75*1.5,3.5*1.5){5};
\node at (4.25*1.5,3.5*1.5){2};
\node at (4.75*1.5,3.5*1.5){3};
\node at (5.25*1.5,3.5*1.5){2};
\node at (5.75*1.5,3.5*1.5){11};
\node at (0.25*1.5,4*1.5){5};
\node at (0.75*1.5,4*1.5){4};
\node at (1.25*1.5,4*1.5){16};
\node at (1.75*1.5,4*1.5){2};
\node at (2.25*1.5,4*1.5){3};
\node at (2.75*1.5,4*1.5){2};
\node at (3.25*1.5,4*1.5){10};
\node at (3.75*1.5,4*1.5){4};
\node at (4.25*1.5,4*1.5){13};
\node at (4.75*1.5,4*1.5){2};
\node at (5.25*1.5,4*1.5){3};
\node at (5.75*1.5,4*1.5){2};
\node at (0.25*1.5,4.5*1.5){2};
\node at (0.75*1.5,4.5*1.5){15};
\node at (1.25*1.5,4.5*1.5){5};
\node at (1.75*1.5,4.5*1.5){11};
\node at (2.25*1.5,4.5*1.5){2};
\node at (2.75*1.5,4.5*1.5){3};
\node at (3.25*1.5,4.5*1.5){2};
\node at (3.75*1.5,4.5*1.5){17};
\node at (4.25*1.5,4.5*1.5){5};
\node at (4.75*1.5,4.5*1.5){4};
\node at (5.25*1.5,4.5*1.5){2};
\node at (5.75*1.5,4.5*1.5){3};
\node at (0.25*1.5,5*1.5){3};
\node at (0.75*1.5,5*1.5){2};
\node at (1.25*1.5,5*1.5){4};
\node at (1.75*1.5,5*1.5){7};
\node at (2.25*1.5,5*1.5){5};
\node at (2.75*1.5,5*1.5){2};
\node at (3.25*1.5,5*1.5){3};
\node at (3.75*1.5,5*1.5){2};
\node at (4.25*1.5,5*1.5){11};
\node at (4.75*1.5,5*1.5){6};
\node at (5.25*1.5,5*1.5){10};
\node at (5.75*1.5,5*1.5){2};
\node at (0.25*1.5,5.5*1.5){2};
\node at (0.75*1.5,5.5*1.5){3};
\node at (1.25*1.5,5.5*1.5){2};
\node at (1.75*1.5,5.5*1.5){10};
\node at (2.25*1.5,5.5*1.5){4};
\node at (2.75*1.5,5.5*1.5){9};
\node at (3.25*1.5,5.5*1.5){2};
\node at (3.75*1.5,5.5*1.5){3};
\node at (4.25*1.5,5.5*1.5){2};
\node at (4.75*1.5,5.5*1.5){5};
\node at (5.25*1.5,5.5*1.5){4};
\node at (5.75*1.5,5.5*1.5){14};
\end{tikzpicture}
\end{center}
\caption{A $24\times 24$ pattern \cite{SO2010}.}
\end{figure}
\begin{prop}
$\chi^{3,3}_{\rho}(\mathbb{Z}^{2})\le 33$.
\end{prop}
\begin{proof}
In the pattern of Figure 6, $B_{i}$ denotes the set of vertices colored by $i$.
Note that $B_{2}$ and $B_{3}$ are both 3-packings.
It can be seen that $B_{16}\cup B_{17}$ form a 11-packing and that four 7-packings form a partition of $B_{2}$ or $B_{3}$.
In order to color $\mathbb{Z}^{2}$ starting with 3, we partition $B_{1}$ into sixteen $i$-packings, $2\le i\le 17$ (since $B_{1}$ is $\bigcup\limits_{i=2}^{17}B_{i}$ translated by the vector $(1,0)$). Let $B'_{i}$ denote a copy of $B_{i}$ translated by $(1,0)$. We use two colors 3 to color $B_{2}$ and $B_{3}$, and one color $i$ in order to color $B_{i}$ for $i\in[4,8]$. We color $B'_{i}$ by one color $i$, for $i\in[3,8]$ and $B'_{2}$ that is a 3-packing is colored by one color 4, one color 5, one color 6 and one color 7. We use the remaining color 8 to color $B_{9}$. We use two colors 9 in order to color $B_{16}$, $B'_{16}$, $B_{17}$ and $B'_{17}$. The remaining color 9 is used to color $B'_{9}$. We use two colors $i$ in order to color $B_{i}$ and $B'_{i}$ for $i\in[10,13]$. The remaining colors 10, 11 ,12 and 13 are used to color $B_{14}$, $B'_{14}$, $B_{15}$ and $B'_{15}$.
\end{proof}
\begin{prop}
$\chi^{3,4}_{\rho}(\mathbb{Z}^{2})\le 20$, $\chi^{3,5}_{\rho}(\mathbb{Z}^{2})\le 17$, $\chi^{3,6}_{\rho}(\mathbb{Z}^{2})\le 14$ and $\forall n\ge 8$, $\chi^{3,n}_{\rho}(\mathbb{Z}^{2})= 8$.
\end{prop}
\begin{proof}
Using Proposition~\ref{sub5}, we define a $(3,n)$-packing coloring of $\mathbb{Z}^{2}$ for each $n=4,5,6$ and $n\ge8$.
$\mathbb{Z}^{2}$ can be partitioned into eight 3-packings, hence $\forall n\ge 8$, $\chi^{3,n}_{\rho}(\mathbb{Z}^{2})= 8$. The other colorings are described in Table~\ref{tabc5}.
\end{proof}
\begin{prop}
$\chi^{4,4}_{\rho}(\mathbb{Z}^{2})\le 56$, $\chi^{4,5}_{\rho}(\mathbb{Z}^{2})\le 34$, $\chi^{4,6}_{\rho}(\mathbb{Z}^{2})\le 28$ and $\forall n\ge 13$, $\chi^{4,n}_{\rho}(\mathbb{Z}^{2})= 13$.
\end{prop}
\begin{proof}
Using Proposition~\ref{sub6}, we define a $(4,n)$-packing coloring of $\mathbb{Z}^{2}$ for each $n=4,5,6$ and $n\ge13$.
$\mathbb{Z}^{2}$ can be partitioned into thirteen 4-packings, hence $\forall n\ge 13$ $\chi^{4,n}_{\rho}(\mathbb{Z}^{2})= 13$. The other colorings are described in Table~\ref{tabc3}.
\end{proof}
\subsection{Triangular lattice}
\begin{prop}
$\chi^{1,1}_{\rho}(\mathscr{T})=\infty, \ \chi^{3,2}_{\rho}(\mathscr{T})=\infty, \ \chi^{4,3}_{\rho}(\mathscr{T})=\infty,\ \chi^{5,4}_{\rho}(\mathscr{T})=\infty, \ \chi^{7,5}_{\rho}(\mathscr{T})=\infty\ and\  \chi^{8,6}_{\rho}(\mathscr{T})=\infty.$
\end{prop}
\begin{proof}
Let $\mathscr{T}$ be the triangular lattice and $k$ be an integer, $k\ge8$.\newline
$\sum\limits_{i=1}^{k}\frac{1}{A(i)}=\sum\limits_{i=1}^{k}\frac{1}{A(2i)}+\sum\limits_{i=0}^{k}\frac{1}{A(2i+1)}=\sum\limits_{i=1}^{k}\frac{1}{3i^{2}+3i +1}+\sum\limits_{i=0}^{k}\frac{1}{3i^{2}+6i+3}<\frac{1}{3}\sqrt{3}\pi\tanh(\frac{1}{6}\pi\sqrt{3})+\frac{1}{18}\pi^{2}-1<0.854.$\newline
Therefore: $\sum\limits_{i=1}^{k}\frac{1}{A(i)}<0.854<1$, 
$\sum\limits_{i=3}^{k}\frac{2}{A(i)}<2(0.854-\sum\limits_{i=1}^{2}\frac{1}{A(i)})<0.755<1$, 
$\sum\limits_{i=4}^{k}\frac{3}{A(i)}<3(0.854-\sum\limits_{i=1}^{3}\frac{1}{A(i)})<0.883<1$, 
$\sum\limits_{i=5}^{k}\frac{4}{A(i)}<4(0.854-\sum\limits_{i=1}^{4}\frac{1}{A(i)})<0.966<1$, 
$\sum\limits_{i=7}^{k}\frac{5}{A(i)}<5(0.854-\sum\limits_{i=1}^{6}\frac{1}{A(i)})<0.887<1$ and 
$\sum\limits_{i=8}^{k}\frac{6}{A(i)}<6(0.854-\sum\limits_{i=1}^{7}\frac{1}{A(i)})<0.940<1$.
Corollary 2.3 allows us to conclude.
\end{proof}
\begin{prop}
$\chi^{1,2}_{\rho}(\mathscr{T})\le 6$ and $\forall n\ge 3$, $\chi^{1,n}_{\rho}(\mathscr{T})= 3$.
\end{prop}
\begin{proof}
Using Proposition~\ref{sub7}, we define a $(1,n)$-packing coloring of $\mathscr{T}$ for each $n=2$ and $n\ge3$.
$\mathscr{T}$ can be partitioned into three independent sets, hence $\forall n\ge 3$, $\chi^{1,n}_{\rho}(\mathscr{T})= 3$. The other coloring is described in the following table.
\begin{center}
\begin{tabular}{| c | c | c | c |}
\hline
\multirow{3}{*}{(1,2)-packing} & \multirow{2}{*}{P} & \multirow{2}{*}{2$\times X_{1}$} & $X_{1}$ \\ \cline{4-4}
& & & 4$\times A_{3}$ \\ \cline{2-4}
& C & 2$\times$1 & 2$\times$2, 2$\times$3 \\ \hline 
\end{tabular}
\end{center}
\end{proof}
\begin{prop}
$\chi^{2,4}_{\rho}(\mathscr{T})\le 16$, $\chi^{2,5}_{\rho}(\mathscr{T})\le 13$, $\chi^{2,6}_{\rho}(\mathscr{T})\le 10$ and $\forall n\ge 7$, $\chi^{2,n}_{\rho}(\mathscr{T})= 7$.
\end{prop}
\begin{proof}
Using Proposition~\ref{sub8}, we define a $(2,n)$-packing coloring of $\mathscr{T}$ for each $n=4,5,6$ and $n\ge7$.
$\mathscr{T}$ can be partitioned into seven 2-packings, hence $\forall n\ge 7$, $\chi^{2,n}_{\rho}(\mathscr{T})= 7$. The other colorings are described in Table~\ref{tabc6}.
\end{proof}
\begin{prop}
$\chi^{3,4}_{\rho}(\mathscr{T})\le 72$, $\chi^{3,5}_{\rho}(\mathscr{T})\le 38$, $\chi^{3,6}_{p}(\mathscr{T})\le 26$ and $\forall n\ge 12$, $\chi^{3,n}_{\rho}(\mathscr{T})= 12$.
\end{prop}
\begin{proof}
Using Proposition~\ref{sub9}, we define a $(3,n)$-packing coloring of $\mathscr{T}$ for each $n=4,5,6$ and $n\ge12$.
$\mathscr{T}$ can be partitioned into twelve 3-packings, hence $\forall n\ge 12$, $\chi^{3,n}_{\rho}(\mathscr{T})= 12$. The other colorings are described in Table~\ref{tabc7}.
\end{proof}
\section{Conclusion}
We have determined or bounded the $(d,n)$-packing chromatic number of three lattices $\mathscr{H}$, $\mathbb{Z}^{2}$ and $\mathscr{T}$ for small values of $d$ and $n$. 
Further studies can be done with other values of $d$ and $n$ or improving existing values.
The $(d,n)$-packing chromatic number can also be investigated for other lattices. As an example, we can prove, using color patterns defined in \cite{TO2010} for distance graphs, that for the octagonal lattice $\mathscr{O}$, \textit{i.e} the strong product of two infinite path (which is a supergraph of $\mathscr{T}$), \textit{$\chi^{1,2}_{\rho}(\mathscr{O})\le 58$}. 
For other finite or infinite graphs, like $k$-regular infinite trees, the method has to be adapted or changed since a maximized packing cannot be described as easily as those considered in this paper.
Also, for each of three lattices studied, finding a sequence $S$ such that $\chi^{S}_{\rho}=k$ and there is no $S$-packing $k$ coloring where the $s_1$-packing is maximized could be an interesting result.

\begin{appendices}
\section{Distances in the three lattices}
\begin{defn}[\cite{JA2005}]
Let $v=(a,b)$ be a vertex in the hexagonal lattice. Then the type of $v$ is
$$\tau(v)=a+b+1\pmod 2.$$
\end{defn}
As $\mathscr{H}=V_{0}\cup V_{1}$ is a bipartite graph, the type of a vertex $v$ corresponds to the index of the set $V_{i}$ to which $v$ belongs. 
\begin{prop}[\cite{JA2005}]
Let $v_{1}=(a_{1},b_{1})$,  $v_{2}=(a_{2},b_{2})$ be two vertices of the hexagonal lattice and assume that $b_{1}\ge b_{2}$. Then the distance between $v_{1}$ and $v_{2}$ is
$$
d(v_{1},v_{2})=\left\{
    \begin{array}{ll}
        |a_{1}-a_{2}|+|b_{1}-b_{2}| & \mbox{if\ }|a_{1}-a_{2}|\ge|b_{1}-b_{2}|; \\
        2|b_{1}-b_{2}|-\tau(v_{1})+\tau(v_{2})  & \mbox{if\ }|a_{1}-a_{2}|<|b_{1}-b_{2}|.
    \end{array}
\right.
$$
\end{prop}
\begin{exmp}
The set $X_{2}$ from Figure 2 is a 2-packing in $\mathscr{H}$.
\end{exmp}
\begin{proof}
Let $x$ and $y$ be integers, then\newline
$d((2(x+1)+4y,x+1),(2x+4y,x))=|2x+4y+2-2x-4y|+|x+1-x|=3>2$ and
$d((2x+4(y+1),x),(2x+4y,x))=4>2$;\newline
let $i$ and $j$ be integers, then $d((2(x+i)+4(y+j),x+i),(2x+4y,x))\ge min(d((2(x+1)+4y,x+1),(2x+4y,x)),d((2x+4(y+1),x),(2x+4y,x)))=3$, hence $X_{2}$ is a 2-packing.
\end{proof}
\begin{claim}
Let $v_{1}=(a_{1},b_{1})$ and $v_{2}=(a_{2},b_{2})$ be two vertices of the square lattice. Then the distance between $v_{1}$ and $v_{2}$ is
$$
d(v_{1},v_{2})=|a_{1}-a_{2}|+|b_{1}-b_{2}|.
$$
\end{claim}
\begin{exmp}
The set $X_{2}$ from Figure 4 is a 2-packing in $\mathbb{Z}^{2}$.
\end{exmp}
\begin{proof}
Let $x$ and $y$ be integers, then\newline
$d((2(x+1)+y,x+1+3y),(2x+y,x+3y))=|2x+y+2-2x-y|+|x+1+3y-x-3y|=3>2$ and
$d((2x+y+1,x-1+3(y+1),(2x+y,x+3y))=4>2$,
to conclude $X_{2}$ is a 2-packing.
\end{proof}
\begin{claim}
Let $v_{1}=(a_{1},b_{1})$ and $v_{2}=(a_{2},b_{2})$ be two vertices of the triangular lattice. Then the distance between $v_{1}$ and $v_{2}$ is
$$
d(v_{1},v_{2})=\left\{
    \begin{array}{ll}
        max(|a_{1}-a_{2}|,|b_{1}-b_{2}|) & \mbox{if\ }((a_{1}\ge a_{2})\land(b_{1}\le b_{2}))\lor ((a_{1}\le a_{2})\land(b_{1}\ge b_{2})) ; \\
        |a_{1}-a_{2}|+|b_{1}-b_{2}|  & \mbox{otherwise}.
    \end{array}
\right.
$$
\end{claim}
\begin{exmp}
The set $X_{1}$ from Figure 5 is an independent set in $\mathscr{T}$.
\end{exmp}
\begin{proof}
Let $x$ and $y$ be integers, then,\newline
$d((x+1+3y,x+1),(x+3y,x))=|x+1+3y-x-3y|+|x+1-x|=2>1$ and
$d((x+3(y+1),x),(x+3y,x))=3>1$, 
to conclude $X_{1}$ is an independent set.
\end{proof}
\section{Decomposition of an $i$-packing in the three lattices}
\renewcommand\thetable{\thesection.\arabic{table}}
\begin{table}[H]
\begin{center}
\begin{tabular}{| c |  c | c | c |}
\hline
$i$ & Number of & Description of & Family of\\
& $i$-packings & a $i$-packing & translation vectors \\ \hline
$4k-1$ & $k^{2}$  & $\{3kx+6ky,kx)\}$ & $(3i+6j,i)$ \\
 & & & $i,j\in\{0,\ldots,k-1\}$\\ \hline
$6k-1$ & $3k^{2}$  & $\{3kx+6ky,3kx)\}$ & $(3i+6j,3i+2a)$ \\
 & & & $i,j\in\{0,\ldots,k-1\}$, $a\in\{0,1,2\}$\\ \hline
$10k-1$ & $8k^{2}$  & $\{6kx+12ky,4kx)\}$ & $(6i+12j+3b,4i+2a+b)$ \\
 & & & $i,j\in\{0,\ldots,k-1\}$,\\ 
 & & & $a\in\{0,1,2,3\}$, $b\in\{0,1\}$\\ \hline
$18k-1$ & $24k^{2}$  & $\{12kx+24ky,6kx)\}$ & $(12i+24j+3b,6i+2a+b)$ \\
 & & & $i,j\in\{0,\ldots,k-1\}$,\\ 
 & & & $a\in\{0,\ldots,5\}$, $b\in\{0,1,2,3\}$\\ \hline
\end{tabular}
\caption{Decomposition of $X_{3}$ in $\mathscr{H}$ into $i$-packings.}
\label{tab4}
\end{center}
\end{table}
\begin{table}[h]
\begin{center}
\begin{tabular}{| c |  c | c | c |}
\hline
$i$ & Number of & Description of & Family of\\
& $i$-packings & a $i$-packing & translation vectors \\ \hline
$5k-1$ & $k^{2}$  & $\{3kx-ky,2kx+3ky)\}$ & $(3i-j,2i+3j)$ \\
 & & & $i,j\in\{0,\ldots,k-1\}$\\ \hline
$6k-1$ & $2k^{2}$  & $\{7kx-ky,kx+3ky)\}$ & $(7i+3a-j,i+2a+3j)$ \\
 & & & $i,j\in\{0,\ldots,k-1\}$, $a\in\{0,1\}$\\ \hline
$8k-1$ & $3k^{2}$  & $\{7kx+2ky,kx+5ky)\}$ & $(7i+2j+3a,i+5j+2a)$ \\
 & & & $i,j\in\{0,\ldots,k-1\}$, $a\in\{0,1,2\}$\\ \hline
$11k-1$ & $6k^{2}$  & $\{-2kx+11ky,6kx)\}$ & $(-2i+11j+7a,6i+a)$ \\
 & & & $i,j\in\{0,\ldots,k-1\}$, $a\in\{0,\ldots,5\}$\\ \hline
\end{tabular}
\caption{Decomposition of $X_{4}$ in $\mathscr{H}$ into $i$-packings.}
\label{tab5}
\end{center}
\end{table}
\begin{table}[h]
\begin{center}
\begin{tabular}{| c | c |  c | c | c |}
\hline
& $i$ & Number of & Description of & Family of\\
& & $i$-packings & a $i$-packing & translation vectors \\ \hline
\multirow{4}{*}{$X_{2}$} & $3k-1$ & $k^{2}$  & $\{2kx-ky,kx+2ky)\}$ & $(2i-j,i+2j)$ \\
& & & & $i,j\in\{0,\ldots,k-1\}$\\ \cline{2 - 5}
& $4k-1$ & $2k^{2}$  & $\{4kx+ky,2kx+3ky)\}$ & $(4i+2a+j,2i+2a+3j)$ \\
& & & & $i,j\in\{0,\ldots,k-1\}$, $a\in\{0,1\}$\\ \hline \hline
\multirow{2}{*}{$X_{3}$} & $4k-1$ & $k^{2}$  & $\{2kx+4ky,2kx)\}$ & $(2i+4j,2i)$ \\
& & & & $i,j\in\{0,\ldots,k-1\}$\\ \hline \hline
\multirow{4}{*}{$X_{4}$} & $5k-1$ & $k^{2}$  & $\{3kx-2ky,2kx+3ky)\}$ & $(3i-2j,2i+3j)$ \\
& & & & $i,j\in\{0,\ldots,k-1\}$\\ \cline{2 - 5}
& $6k-1$ & $2k^{2}$  & $\{6kx+ky,4kx+5ky)\}$ & $(6i+j+3a,4i+5j+2a)$ \\
& & & & $i,j\in\{0,\ldots,k-1\}$, $a\in\{0,1\}$\\ \hline
\end{tabular}
\caption{Decomposition of $X_{2}$, $X_{3}$ and $X_{4}$ in $\mathbb{Z}^{2}$ into $i$-packings.}
\label{tab7}
\end{center}
\end{table}
\begin{table}[h]
\begin{center}
\begin{tabular}{| c | c |  c | c | c |}
\hline
& $i$ & Number of & Description of & Family of\\
& & $i$-packings & a $i$-packing & translation vectors \\ \hline
\multirow{4}{*}{$X_{1}$} & $2k-1$ & $k^{2}$  & $\{kx+3ky,kx)\}$ & $(i+3j,i)$ \\
& & & & $i,j\in\{0,\ldots,k-1\}$\\ \cline{2 - 5}
& $3k-1$ & $3k^{2}$  & $\{3kx+3ky,3kx)\}$ & $(3i+3j+a,3i+a)$ \\
& & & & $i,j\in\{0,\ldots,k-1\}$, $a\in\{0,1,2\}$\\ \hline \hline
\multirow{2}{*}{$X_{2}$} & $3k-1$ & $k^{2}$  & $\{2kx+7ky,kx)\}$ & $(2i+7j,i)$ \\
& & & & $i,j\in\{0,\ldots,k-1\}$\\ \hline \hline
\multirow{4}{*}{$X_{3}$} & $4k-1$ & $k^{2}$  & $\{2kx+6ky,2kx)\}$ & $(2i+6j,2i)$ \\
& & & & $i,j\in\{0,\ldots,k-1\}$\\ \cline{2 - 5}
& $6k-1$ & $3k^{2}$  & $\{6kx+6ky,6kx)\}$ & $(6i+6j+2a,6i+2a)$ \\
& & & & $i,j\in\{0,\ldots,k-1\}$, $a\in\{0,1,2\}$\\ \hline
\end{tabular}
\caption{Decomposition of $X_{1}$, $X_{2}$ and $X_{3}$ in $\mathscr{T}$ into $i$-packings.}
\label{tab10}
\end{center}
\end{table}
\section{Decomposition and associated colors}
\begin{table}[h]
\begin{center}
\begin{tabular}{| c | c |  c | c | c | c | c |}
\hline
\multirow{6}{*}{(3,2)-packing} & \multirow{3}{*}{P} & \multirow{3}{*}{2$\times X_{3}$} & $X_{3}$ & $X_{3}$ & $X_{3}$ & $X_{3}$ \\ \cline{4-7}
& & & 3$\times X_{5}$ & 4$\times X_{7}$ & 4$\times X_{9}$, 8$\times X_{15}$ & $X_{5}$, 3$\times X_{11}$, \\ 
& & & &  & & 4$\times X_{17}$, 6$\times X_{23}$\\ \cline{2-7}
& \multirow{3}{*}{C} & 2$\times$3 & 2$\times$4, 5 & 2$\times$6, 2$\times$7 & 2$\times$8, 2$\times 9$, & 5, 2$\times10$, 2$\times 11$, \\
& & & & &  2$\times 12$, 2$\times 13$, & 2$\times 16$, 2$\times 17$, 2$\times 18$, \\
& & & & & 2$\times 14$, 2$\times 15$ & 2$\times 19$, $20$  \\ \hline
\multirow{3}{*}{(3,3)-packing} & \multirow{2}{*}{P} & \multirow{2}{*}{3$\times X_{3}$} & $X_{3}$ & $X_{3}$ & $X_{3}$ & \\ \cline{4-7}
& & & 3$\times X_{5}$ & 3$\times X_{5}$ & 4$\times X_{7}$ & \\ \cline{2-7}
& C & 3$\times$3 & 3$\times$4 & 3$\times$5 & 3$\times$6, 7 & \\ \hline
\multirow{3}{*}{(3,4)-packing} & \multirow{2}{*}{P} & \multirow{2}{*}{4$\times X_{3}$} & $X_{3}$ & $X_{3}$ & & \\ \cline{4-7}
& & & 3$\times X_{5}$ & 3$\times X_{5}$ & & \\\cline{2-7}
& C & 4$\times$3 & 3$\times$4 & 4, 2$\times$5 & & \\ \hline
\multirow{3}{*}{(3,5)-packing} & \multirow{2}{*}{P} & \multirow{2}{*}{5$\times X_{3}$} & $X_{3}$ & & & \\ \cline{4-7}
& & & 3$\times X_{5}$ & & & \\ \cline{2-7}
& C & 5$\times$3 & 3$\times$4 & & & \\ \hline
\end{tabular}
\end{center}
\caption{Decomposition of $\mathscr{H}$ into $3$-packings and associated colors.}
\label{tabc1}
\end{table}
\begin{table}[h]
\begin{center}
\begin{tabular}{| c | c | c | c | c | c | c |}
\hline
\multirow{7}{*}{(4,3)-packing} & \multirow{2}{*}{P} & \multirow{2}{*}{3$\times X_{4}$} & $X_{4}$ & 2$\times X_{4}$ & 2$\times X_{4}$ & $X_{4}$ \\ \cline{4-7}
& & & 2$\times A_{5}$ & 6$\times A_{7}$ & $A_{5}$, 6$\times A_{9}$ & 6$\times A_{11}$, 4$\times A_{19}$\\ \cline{2-7}
& C & 3$\times$4 & 2$\times$5 & 3$\times$6, 3$\times$7 & 5, 3$\times$8, 3$\times 9$ & 3$\times$10, 3$\times$11, 18, 3$\times$19 \\ \cline{2-7}
& \multirow{2}{*}{P} & & & & $X_{4}$ & $X_{4}$ \\ \cline{3-7}
& & & & & 9$\times A_{14}$ & 11$\times A_{19}$, 10$\times A_{23}$ \\ \cline{2-7}
& \multirow{2}{*}{C} & & & & 3$\times 12$, 3$\times 13$, & 3$\times 15$, 3$\times 16$, 3$\times 17$, \\
& & & & & 3$\times 14$ & 2$\times 18$, 3$\times 20$, 3$\times 21$ \\
& & & & & & 3$\times 22$, 23 \\ \hline
\multirow{3}{*}{(4,4)-packing} & \multirow{2}{*}{P} & \multirow{2}{*}{4$\times X_{4}$} & 2$\times X_{4}$ & 2$\times X_{4}$ & 2$\times X_{4}$ & $X_{4}$ \\ \cline{4-7}
& & & 4$\times A_{5}$ & 6$\times A_{7}$ & 8$\times A_{9}$ & 2$\times A_{7}$, 3$\times A_{14}$ \\ \cline{2-7}
& C & 4$\times$4 & 4$\times$5 & 4$\times$6, 2$\times$7 & 4$\times$8, 4$\times$9 & 2$\times$7, 3$\times 10$ \\ \hline
\multirow{3}{*}{(4,5)-packing} & \multirow{2}{*}{P} & \multirow{2}{*}{5$\times X_{4}$} & 2$\times X_{4}$ & 3$\times X_{4}$ & $X_{4}$ & \\ \cline{4-7}
& & & 4$\times A_{5}$ & 9$\times A_{7}$ & $A_{5}$, 2$\times A_{9}$ & \\ \cline{2-7}
& C & 5$\times$4 & 4$\times$5 & 5$\times$6, 4$\times$7 & 5, 7, 8 & \\ \hline 
\multirow{3}{*}{(4,6)-packing} & \multirow{2}{*}{P} & \multirow{2}{*}{6$\times X_{4}$} & 3$\times X_{4}$ & 2$\times X_{4}$ & & \\ \cline{4-7}
& & & 6$\times A_{5}$ & 6$\times A_{7}$ & & \\ \cline{2-7}
& C & 6$\times$4 & 6$\times$5 & 6$\times$6 & & \\ \hline 
\end{tabular}
\end{center}
\caption{Decomposition of $\mathscr{H}$ into $4$-packings and associated colors.}
\label{tabc2}
\end{table}
\begin{table}[h]
\begin{center}
\begin{tabular}{| c | c | c | c | c | c |}
\hline
\multirow{4}{*}{(2,2)-packing} & \multirow{2}{*}{P} & \multirow{2}{*}{2$\times X_{2}$} & $X_{2}$ & $X_{2}$ & $X_{2}$ \\ \cline{4-6}
& & & 2$\times A_{3}$ & 4$\times A_{5}$ & 6$\times A_{8}$,  6$\times A_{11}$ \\ \cline{2-6}
& \multirow{2}{*}{C} & 2$\times$2 & 2$\times$3 & 2$\times$4, 2$\times$5 & 2$\times$6, 2$\times$7, 2$\times$8 \\ 
& & & & & 2$\times$9, 2$\times$10, 2$\times$11 \\ \hline
\multirow{3}{*}{(2,3)-packing} & \multirow{2}{*}{P} & \multirow{2}{*}{3$\times X_{2}$} & $X_{2}$ & $X_{2}$ &\\ \cline{4-6}
& & & 2$\times A_{3}$ & $A_{3}$, 2$\times A_{5}$ & \\ \cline{2-6}
& C & 3$\times$2 & 2$\times$3 & 3, 2$\times$4 & \\ \hline
\multirow{3}{*}{(2,4)-packing} & \multirow{2}{*}{P} & \multirow{2}{*}{4$\times X_{2}$} & $X_{2}$ & &\\ \cline{4-6}
& & & 2$\times A_{3}$ & & \\ \cline{2-6}
& C & 4$\times$2 & 2$\times$3 & & \\ \hline
\end{tabular}
\end{center}
\caption{Decomposition of $\mathbb{Z}^{2}$ into $2$-packings and associated colors.}
\label{tabc4}
\end{table}
\begin{table}[h]
\begin{center}
\begin{tabular}{| c | c | c | c | }
\hline
\multirow{4}{*}{(3,4)-packing} & \multirow{2}{*}{P} & \multirow{2}{*}{4$\times X_{3}$} & 4$\times X_{3}$ \\ \cline{4-4}
& & & 16$\times A_{7}$  \\ \cline{2-4}
& C & 4$\times$3 & 4$\times$4, 4$\times$5, 4$\times$6, 4$\times$7 \\ \hline
\multirow{3}{*}{(3,5)-packing} & \multirow{2}{*}{P} & \multirow{2}{*}{5$\times X_{3}$} & 3$\times X_{3}$ \\ \cline{4-4}
& & & 12$\times A_{7}$  \\ \cline{2-4}
& C & 5$\times$3 & 5$\times$4, 5$\times$5, 2$\times$6 \\ \hline
\multirow{3}{*}{(3,6)-packing} & \multirow{2}{*}{P} & \multirow{2}{*}{6$\times X_{3}$} & 2$\times X_{3}$ \\ \cline{4-4}
& & & 8$\times A_{7}$ \\ \cline{2-4}
& C & 6$\times$3 & 6$\times$4, 2$\times$5 \\ \hline
\end{tabular}
\end{center}
\caption{Decomposition of $\mathbb{Z}^{2}$ into $3$-packings and associated colors.}
\label{tabc5}
\end{table}
\begin{table}[h]
\begin{center}
\begin{tabular}{| c | c | c | c | c | c | c | c |}
\hline
\multirow{4}{*}{(4,4)-packing} & \multirow{2}{*}{P} & \multirow{2}{*}{4$\times X_{4}$} & 2$\times X_{4}$ &  4$\times X_{4}$ & $X_{4}$ & $X_{4}$ & $X_{4}$ \\ \cline{4-8}
& & & 4$\times A_{5}$ & 16$\times A_{9}$ & 8$\times A_{11}$ & 9$\times A_{14}$ & 3$\times A_{14}$, 12$\times A_{17}$ \\ \cline{2-8}
& \multirow{2}{*}{C} & 4$\times$4 & 4$\times$5 & 4$\times$6, 4$\times$7, & 4$\times$10, & 4$\times$12 & 3$\times$14, 4$\times$15, \\
& & & & 4$\times$8, 4$\times$9 & 4$\times$11 & 4$\times$13, 14 & 4$\times$16, 4$\times$17 \\ \cline{2-8}
\multirow{6}{*}{(4,5)-packing} & \multirow{3}{*}{P} & \multirow{2}{*}{5$\times X_{4}$} & 2$\times X_{4}$ & 5$\times X_{4}$ & $X_{4}$ & & \\ \cline{4-8}
& & & 4$\times A_{5}$ & $A_{5}$, 18$\times A_{9}$ &2$\times A_{9}$, & &\\ 
& & & & &  4$\times A_{11}$& &\\ \cline{2-8}
& \multirow{2}{*}{C} & 5$\times$4 & 4$\times$5 & 5, 5$\times$6, 5$\times$7 & 2$\times$9, & & \\
& & & & 5$\times$8, 3$\times$9 & 4$\times$10 & & \\ \hline
\multirow{3}{*}{(4,6)-packing} & \multirow{2}{*}{P} & \multirow{2}{*}{6$\times X_{4}$} & 3$\times X_{4}$ & 4$\times X_{4}$ & & & \\ \cline{4-8}
& & & 6$\times A_{5}$ & 16$\times A_{9}$ & & & \\ \cline{2-8}
& C & 6$\times$4 & 6$\times$5 & 6$\times$6, 6$\times$7, 4$\times$8 & & & \\ \hline
\end{tabular}
\end{center}
\caption{Decomposition of $\mathbb{Z}^{2}$ into $4$-packings and associated colors.}
\label{tabc3}
\end{table}
\begin{table}[h]
\begin{center}
\begin{tabular}{| c | c | c | c |}
\hline
\multirow{3}{*}{(2,4)-packing} & \multirow{2}{*}{P} & \multirow{2}{*}{4$\times X_{2}$} & 3$\times X_{2}$ \\ \cline{4-4}
& & & 12$\times A_{5}$ \\ \cline{2-4}
& C & 4$\times$2 & 4$\times$3, 4$\times$4, 4$\times$5 \\ \hline 
\multirow{3}{*}{(2,5)-packing} & \multirow{2}{*}{P} & \multirow{2}{*}{5$\times X_{2}$} & 2$\times X_{2}$ \\ \cline{4-4}
& & & 8$\times A_{5}$ \\ \cline{2-4}
& C & 5$\times$2 & 5$\times$3, 3$\times$4 \\ \hline 
\multirow{3}{*}{(2,6)-packing} & \multirow{2}{*}{P} & \multirow{2}{*}{6$\times X_{2}$} & $X_{2}$ \\ \cline{4-4}
& & & 4$\times A_{5}$ \\ \cline{2-4}
& C & 6$\times$2 & 4$\times$3 \\ \hline 
\end{tabular}
\end{center}
\caption{Decomposition of $\mathscr{T}$ into $2$-packings and associated colors.}
\label{tabc6}
\end{table}
\begin{table}[h]
\begin{center}
\begin{tabular}{| c | c | c | c | c | c |}
\hline
\multirow{6}{*}{(3,4)-packing} & \multirow{2}{*}{P} & \multirow{2}{*}{4$\times X_{3}$} & 2$\times X_{3}$ & 2$\times X_{3}$ & 2$\times X_{3}$ \\ \cline{4-6}
& & & 6$\times A_{5}$ & 8$\times A_{7}$ & 2$\times A_{5}$, 12$\times A_{11}$ \\ \cline{2-6}
& C & 4$\times$3 & 4$\times$4, 2$\times$5 & 4$\times$6, 4$\times$7 & 2$\times$5, 4$\times$9, 4$\times$10, 4$\times$11 \\ \cline{2-6} 
 & \multirow{2}{*}{P} &  &  & $X_{3}$ & $X_{3}$ \\ \cline{4-6}
& & & & 16$\times A_{15}$ & 4$\times A_{11}$, 20$\times A_{23}$ \\ \cline{2-6}
& \multirow{2}{*}{C} & & & 4$\times$12, 4$\times$13 & 4$\times$8, 4$\times$16, 4$\times$17, \\ 
& & & & 4$\times$14, 4$\times$15 & 4$\times$18, 4$\times$19, 4$\times$20 \\ \hline 
\multirow{3}{*}{(3,5)-packing} & \multirow{2}{*}{P} & \multirow{2}{*}{5$\times X_{3}$} & 3$\times X_{3}$ & 2$\times X_{3}$ & 2$\times X_{3}$ \\ \cline{4-6}
& & & 9$\times A_{5}$ & 8$\times A_{7}$ & 18$\times A_{11}$ \\ \cline{2-6}
& C & 5$\times$3 & 5$\times$4, 4$\times$5 & 5$\times$6, 3$\times$7 & 5, 2$\times$7, 5$\times$8, 5$\times$9, 3$\times$10 \\ \hline 
\multirow{3}{*}{(3,6)-packing} & \multirow{2}{*}{P} & \multirow{2}{*}{6$\times X_{3}$} & 4$\times X_{3}$ & 2$\times X_{3}$ & \\ \cline{4-6}
& & & 12$\times A_{5}$ & 8$\times A_{7}$ & \\ \cline{2-6}
& C & 6$\times$3 & 6$\times$4, 6$\times$5 & 6$\times$6, 2$\times$7 & \\ \hline 
\end{tabular}
\end{center}
\caption{Decomposition of $\mathscr{T}$ into $3$-packings and associated colors.}
\label{tabc7}
\end{table}
\end{appendices}
\end{document}